\documentclass[10pt]{article} 

\usepackage[margin=1.1in]{geometry}
\usepackage{blindtext}
\usepackage{latexsym}
\usepackage{amsmath,amssymb,amsfonts,amsthm}
\usepackage{graphicx,subfigure,psfrag}
\usepackage{mathtools}
\usepackage{multicol,multirow}
\usepackage{algorithm}
\usepackage[noend]{algpseudocode}
\usepackage{arydshln}
\usepackage{wrapfig}
\usepackage{lipsum}
\usepackage{fancybox}
\usepackage{xcolor}
\usepackage[most]{tcolorbox}
\usepackage{kotex}
\usepackage{soul}
\usepackage{empheq}
\usepackage[authoryear, round]{natbib}
\usepackage{booktabs}
\usepackage{hyperref}
\usepackage{cleveref}
\usepackage{enumitem}
\usepackage{mathtools}
\usepackage{bm}
\usepackage{titlesec}
\usepackage{fancyhdr}
\usepackage{mdframed}
\usepackage{microtype}
\usepackage{parskip}

\usepackage{tikz}
\usepackage{pgfplots}
\pgfplotsset{compat=1.18}
\usetikzlibrary{arrows.meta,positioning,calc,fit,backgrounds}

\newtheorem{proposition}{Proposition}

\newtheorem{lemma}{Lemma}

\newtheorem{remark}{Remark}

\hypersetup{
  colorlinks=true,
  linkcolor=blue!70!black,
  citecolor=blue!50!black,
  urlcolor=blue!60!black
}

\providecommand{\E}{\mathbb{E}}
\providecommand{\Prob}{\mathbb{P}}
\newcommand{\R}{\mathbb{R}}
\newcommand{\Var}{\mathrm{Var}}
\newcommand{\cN}{\mathcal{N}}
\newcommand{\cC}{\mathcal{C}}
\newcommand{\cD}{\mathcal{D}}
\newcommand{\cX}{\mathcal{X}}
\newcommand{\dtrain}{\cD_{\mathrm{tr}}}
\newcommand{\dcal}{\cD_{\mathrm{cal}}}
\newcommand{\dtarget}{\cD_{\mathrm{tg}}}
\newcommand{\ps}{p_{\mathrm{s}}}
\newcommand{\pt}{p_{\mathrm{t}}}
\newcommand{\fs}{f_{\mathrm{s}}}
\newcommand{\ft}{f_{\mathrm{t}}}
\newcommand{\bhtheta}{\widehat{\theta}}
\newcommand{\bhmu}{\widehat{\mu}}
\newcommand{\bhsigma}{\widehat{\sigma}}
\newcommand{\sigeps}{\sigma_{\varepsilon}}
\newcommand{\bhsigeps}{\widehat{\sigma}_{\varepsilon}}
\newcommand{\bhSigma}{\widehat{\Sigma}}

\newcommand{\piLt}{\pi_L^{\mathrm{t}}}
\newcommand{\piUt}{\pi_U^{\mathrm{t}}}
\newcommand{\mut}{\mu^{\mathrm{t}}}
\newcommand{\phist}{\phi_{\mathrm{std}}}
\newcommand{\wt}[1]{\widetilde{#1}}
\newcommand{\dTV}{d_{\mathrm{TV}}}
\newcommand{\ind}{\mathbf{1}}

\newcommand{\norm}[1]{\lVert #1 \rVert}
\newcommand{\backbone}{\varphi}
\providecommand{\bbeta}{\widehat{\beta}}
\providecommand{\bq}{\widehat{q}}

\titleformat{\section}{\large\bfseries\color{blue!60!black}}{{\thesection}}{1em}{}[\titlerule]
\titleformat{\subsection}{\normalsize\bfseries\color{blue!40!black}}{{\thesubsection}}{1em}{}
\titleformat{\subsubsection}{\normalsize\itshape\bfseries}{{\thesubsubsection}}{1em}{}

\begin{document}

\begin{center}
  {\LARGE\bfseries \textsf{Conformal Bayes for Two-Sided Censored Gaussian \\ [5pt]
  Regression Under Label Shift }}\\[2em]
  {\large  Seungjin Choi}\\[0.5em]
  {\normalsize CROID Research and aSSIST University, Korea} 
  \end{center}

\vspace{0.5em}
\noindent\rule{\linewidth}{1.5pt}
\vspace{0.5em}

\begin{abstract}
Prediction under label shift becomes nonstandard when responses are censored. 
In a two-sided censored Gaussian model, latent values below $L$ and above $U$ 
are recorded at the boundary values, so the observed predictive distribution is mixed, 
with atoms at $L$ and $U$ and a continuous density on $(L,U)$.  
In this paper we develop conformal Bayes for this mixed-space setting by combining posterior predictive tilting 
with weighted conformal calibration.  Under a two-sided Tobit Gaussian Bayesian prediction head 
with a Laplace posterior approximation, the tilted predictive distribution has left-atom, interior, and right-atom components, 
with a three-term closed-form normalizer.  
The resulting prediction set is a mixed highest density region that can combine boundary atoms 
with an interior interval and can reduce to atom-only sets under strong censoring.  
The main technical issue is that latent label shift does not directly give an ordinary density ratio on the observed censored scale.  
A latent exponential tilt induces tail-averaged atom weights at the censored boundaries, 
while the interior ratio remains density based.  
This yields a mixed observed-space calibration weight with two atom ratios and one interior density ratio.
The weight corrects the calibration measure, while predictive tilting gives target-adapted mixed-HDR geometry.
Synthetic experiments show that weighted tilted conformal Bayes restores marginal coverage 
with smaller sets than weighted source-score calibration, while revealing a trade-off between marginal coverage 
and component-wise behavior across atoms and interior observations.
\end{abstract}


\section{Introduction}
\label{sec:intro}

Many scientific measurements are reported only within a finite detection range
\citep{HeiselDR2011book, WilliamsJR2020bmc}.
A latent continuous response below a lower limit \(L\) is recorded as
\(\wt{Y}=L\), a response above an upper limit \(U\) is recorded as
\(\wt{Y}=U\), and only values in the interval \((L,U)\) are observed exactly.
This type of censoring appears in pharmacokinetic assays, binding measurements,
solubility measurements near saturation, and toxicological endpoints
\citep{KeizerRJ2015prp}.
It is often accompanied by label shift, where the marginal response distribution
changes between training and deployment
\citep{SaerensM2002neco, LiptonZ2018icml, AzizzadenesheliK2019iclr, AlexandariAM2020icml, GargS2020neurips}.
A model may be trained on compounds with moderate activity, but later deployed
on lead candidates enriched for high activity or on safety screens enriched for
extreme outcomes.  The observed distribution then changes through both the
continuous part of the response and the proportions of left and right censored
observations.

Let $Y^*$ denote the latent response and define
\begin{equation}\label{eq:censoring}
\wt{Y}=T(Y^*)=
\begin{cases}
L, & Y^*\le L,\\
Y^*, & L<Y^*<U,\\
U, & Y^*\ge U,
\end{cases}
\end{equation}
where $T: \R \rightarrow \{L\} \cup (L,U) \cup \{U\}$ is the two-sided censoring map.
Even when $Y^*\mid X=x$ is Gaussian, the observed response $\wt{Y}\mid X=x$ is not Gaussian.  
It is a mixed distribution on
\[
\{L\}\cup(L,U)\cup\{U\},
\]
with two point masses at the detection limits and a continuous density on the interior. 
For a fixed input \(x\), the observed conditional distribution is a mixed measure,
\begin{equation}\label{eq:mixed_measure_intro}
dP_{\wt Y\mid X=x}(y)
=
\pi_L(x)\delta_L(dy)
+
f(y\mid x)\mathbf 1\{L<y<U\}\,dy
+
\pi_U(x)\delta_U(dy),
\end{equation}
where \(\mathbf 1\{\cdot\}\) denotes the indicator function, 
\(\delta_L\) and \(\delta_U\) are point masses at the two censoring boundaries,
\[
\pi_L(x)=\Prob(\wt Y=L\mid X=x),
\qquad
\pi_U(x)=\Prob(\wt Y=U\mid X=x),
\]
are the conditional left- and right-censoring probabilities, and
\(f(y\mid x)\) is the conditional density of the exactly observed response on
the interior interval \((L,U)\).
In the latent Gaussian model, these atom probabilities are induced by the tails
of \(Y^*\mid X=x\), namely
\[
\pi_L(x)=\Prob(Y^*\le L\mid X=x),
\qquad
\pi_U(x)=\Prob(Y^*\ge U\mid X=x),
\]
as shown in Fig. \ref{fig:censored-gaussian-dirac-tikz}.

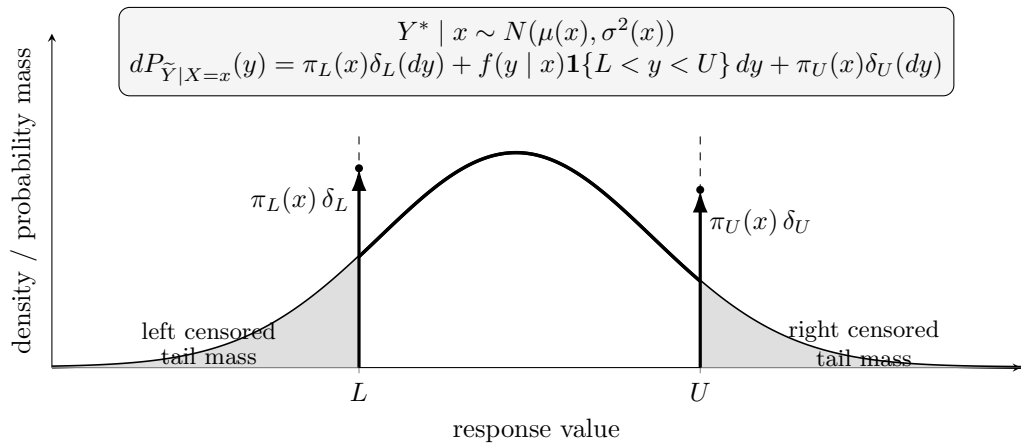
\begin{figure}[ht!]
\centering
\begin{tikzpicture}
\begin{axis}[
    width=0.9\linewidth,
    height=6.0cm,
    xmin=-3.4, xmax=3.7,
    ymin=0, ymax=0.62,
    axis lines=left,
    xlabel={response value},
    ylabel={density / probability mass},
    ytick=\empty,
    xtick={-1.15,1.35},
    xticklabels={$L$,$U$},
    clip=false,
    domain=-3.4:3.7,
    samples=250,
]
\addplot[very thick]
    {1/sqrt(2*pi)*exp(-x^2/2)};
\addplot[
    draw=none,
    fill=gray!25,
    domain=-3.4:-1.15,
]
    {1/sqrt(2*pi)*exp(-x^2/2)} \closedcycle;
\addplot[
    draw=none,
    fill=gray!25,
    domain=1.35:3.7,
]
    {1/sqrt(2*pi)*exp(-x^2/2)} \closedcycle;
\addplot[
    line width=1.3pt,
    domain=-1.15:1.35,
]
    {1/sqrt(2*pi)*exp(-x^2/2)};
\addplot[dashed] coordinates {(-1.15,0) (-1.15,0.43)};
\addplot[dashed] coordinates {(1.35,0) (1.35,0.43)};
\draw[-{Latex[length=3mm,width=2mm]}, line width=1.2pt]
    (axis cs:-1.15,0) -- (axis cs:-1.15,0.37);
\draw[-{Latex[length=3mm,width=2mm]}, line width=1.2pt]
    (axis cs:1.35,0) -- (axis cs:1.35,0.33);
\fill (axis cs:-1.15,0.37) circle (1.5pt);
\fill (axis cs:1.35,0.33) circle (1.5pt);
\node[anchor=east] at (axis cs:-1.15,0.31) {$\pi_L(x)\,\delta_L$};
\node[anchor=west] at (axis cs:1.35,0.27) {$\pi_U(x)\,\delta_U$};
\node[font=\small,align=center] at (axis cs:-2.25,0.045)
{\shortstack{left censored\\tail mass}};
\node[font=\small,align=center] at (axis cs:2.55,0.045)
{\shortstack{right censored\\tail mass}};
\node[
    draw,
    rounded corners,
    fill=gray!8,
    inner sep=4pt,
    align=center,
    anchor=north
] at (rel axis cs:0.5,1.08)
{$Y^*\mid x\sim N(\mu(x),\sigma^2(x))$\\
$dP_{\wt Y\mid X=x}(y)=
\pi_L(x)\delta_L(dy)
+ f(y\mid x)\mathbf{1}\{L<y<U\}\,dy
+ \pi_U(x)\delta_U(dy)$};
\end{axis}
\end{tikzpicture}
\caption{Two-sided censored Gaussian observation model.  A latent Gaussian response
$Y^*\mid x$ is censored at lower and upper detection limits $L$ and $U$.
The left and right latent tail probabilities collapse into two Dirac atoms
$\pi_L(x)\delta_L$ and $\pi_U(x)\delta_U$, while the interior part remains
continuous with density $f(y\mid x)$ on $(L,U)$.}
\label{fig:censored-gaussian-dirac-tikz}
\end{figure}

A classical way to model censored Gaussian responses is Tobit regression \citep{TobinJ1958econ,AmemiyaT1984jecon}.  
In its simplest form, Tobit regression assumes a latent Gaussian response whose value is observed exactly
inside the measurement range and recorded at a boundary outside that range. 
The likelihood combines Gaussian density terms for uncensored observations with Gaussian tail probabilities 
for censored observations.  
We use a Bayesian prediction head based on two-sided Tobit regression as the predictive model, 
but the conformal calibration is carried out on the observed mixed space.

A direct use of standard conformal prediction based on residual scores for \(\wt Y\) 
does not exploit the mixed nature of the observed response.  
The two boundary values are treated as ordinary numerical responses, 
although each boundary value represents an entire latent tail.  
Weighted conformal prediction \citep{TibshiraniR2019neurips} can correct for distribution shift 
once the density ratio on the observed space has been specified.  
Under censoring, this observed-space ratio has three components: a left atom ratio, an interior density ratio,
and a right atom ratio.  For the mixed-HDR score used below, prediction sets also change form.  
They are level sets of a mixed density and may contain one or both boundary atoms together with an interior interval.

Conformal Bayes combines Bayesian predictive modeling with conformal calibration \citep{WassermanL2011ss,FongE2021neurips}.  
The Bayesian model provides a posterior predictive distribution and an adaptive score, 
while the conformal step calibrates this score using held-out data to obtain finite-sample coverage.  
In the uncensored Gaussian setting, split conformal Bayes under label shift can be implemented 
by tilting the Bayesian predictive distribution and using weighted conformal calibration \citep{Choi2026eiml}.  
That construction relies on an ordinary continuous response density.  
The censored case requires a new observed-space formulation because the predictive distribution contains 
both atoms and a continuous component.

We develop conformal Bayes for two-sided censored Gaussian regression under label shift.  
The method combines posterior predictive tilting with weighted conformal calibration on the observed mixed space.  
The Bayesian predictive distribution yields a score that adapts to the target label distribution, 
while the conformal step restores finite-sample marginal coverage 
when the required observed-space importance identity is available.  
The resulting procedures are summarized by three score and weight pairs: unweighted tilted calibration (UT), 
weighted source calibration (WS), and weighted tilted calibration (WT), which are described in Section \ref{subsec:ut_ws_wt}.

The main contributions of this paper are summarized as follows.
\begin{enumerate}[label=(\roman*),leftmargin=2.2em,itemsep=2pt]
\item We derive the source and tilted predictive distributions for a two-sided Tobit Gaussian Bayesian prediction head with a Laplace posterior approximation.  The tilted normalizer decomposes into left-atom, interior, and right-atom terms, which yields closed-form expressions for the mixed predictive law.

\item We identify the main censoring-specific issue in label shift.  A latent label tilt does not induce an ordinary density ratio on the observed censored scale.  Instead, the atom weights are tail-averaged ratios, while the interior weight remains density based.  Marginal atom ratios are therefore approximations unless additional structure makes them exact.

\item We characterize mixed highest density prediction sets on the observed censored space.  Under the Gaussian prediction head, the set is the union of optional boundary atoms and a possibly empty interior interval.  It may include \(L\), \(U\), both atoms, or reduce to atom-only sets under strong censoring.  Disconnected interior geometry can arise only with a multimodal interior predictive density and therefore does not occur under the Gaussian prediction head.

\item We compare unweighted tilted, weighted source, and weighted tilted calibration.  The analysis separates the exact candidate-weighted construction from practical calibration-only plug-in rules based on estimated observed-scale importance weights.  Synthetic experiments show that weighted tilted calibration recovers marginal coverage with smaller sets than weighted source calibration, while coverage across atoms and interior observations can remain uneven.
\end{enumerate}

\section{Related Work}
\label{sec:related}

Conformal prediction provides finite-sample marginal coverage under exchangeability and can be implemented through full 
or split calibration schemes \citep{VovkV2005book,PapadopoulosH2002ecml,ShaferG2008jmlr,BarberRF2023aos,AngelopoulosAN2023ftml}.
A central question is how to retain coverage when calibration and test data are not identically distributed.  Under covariate shift, weighted conformal prediction uses likelihood ratios in the conformal quantile \citep{TibshiraniR2019neurips}.  More general nonexchangeable settings have also been studied through weighted and robust variants \citep{BarberRF2023aos}.  Under label shift, weighted conformal prediction uses important weights in the form of 
marginal label density ratios \citep{LeeHS2025reliableML}. For continuous label shift, recent work develops weighted conformal regression and split conformal Bayes with predictive tilting \citep{Choi2026eiml}.  These methods assume that the response is observed on an ordinary continuous scale.  The present paper asks what remains valid, and what must be changed, when the observed response has atoms at censoring limits.

Bayesian conformal methods use Bayesian predictive distributions as conformal ingredients rather than replacing Bayesian modeling by a purely residual-based procedure.  The broad motivation goes back to the frequentist calibration of Bayesian procedures \citep{WassermanL2011ss}.  Conformal Bayesian computation constructs calibrated Bayesian predictive sets by reweighting posterior samples through an add-one importance scheme \citep{FongE2021neurips}.  Split conformal Bayes under label shift takes a different but complementary route: it tilts the posterior predictive distribution to align the score with a target label distribution and then uses weighted conformal calibration \citep{Choi2026eiml,LeeHS2026eiml}.  Our contribution is not a new principle of conformal Bayes.  It is the mixed-space extension needed when the Bayesian predictive distribution has both Dirac atoms and an interior density.

Censored Gaussian regression has a long statistical history.  The Tobit model was introduced for limited dependent variables \citep{TobinJ1958econ}, and classical surveys describe a family of censored and truncated regression models \citep{AmemiyaT1984jecon}.  
In a two-sided Tobit model \citep{RosettRN1975econ}, uncensored observations contribute Gaussian density terms, while left and right censored observations contribute Gaussian tail probabilities.  Bayesian implementations can use analytic approximations, such as Laplace approximation \citep{TierneyL1986jasa}, or latent variable augmentation schemes related to probit and censored normal models \citep{AlbertJH1993jasa}.  We use a two-sided Tobit Bayesian prediction head with a Laplace approximation because it gives a tractable posterior predictive distribution with explicit atom probabilities and an interior Gaussian density.  The conformal problem then begins after this predictive distribution has been formed.

There is also related work on conformal prediction for censored outcomes, especially survival analysis.  
Conformalized survival analysis and its extensions use censoring weights, imputation, doubly robust correction, 
or sensitivity analysis to handle event times that may be unobserved because of right censoring 
or dependent censoring \citep{CandesE2023jrsssb,JinY2023pnas,GuiY2024biometrika,DavidovH2025iclr,SesiaM2025icml}.  
Recent work by \citet{DavidovH2025iclr} treats general right censored data, while \citet{SesiaM2025icml} develops a doubly robust 
conformalized survival method for right censored data.  
This literature is concerned with missing or partially observed event times and the censoring mechanism that hides them.  
Our setting is different.  The response is a measured quantity subject to lower and upper detection limits, 
so the observed outcome itself is a mixed object with two boundary atoms and one continuous interior component.  
The goal is not to infer an unobserved event time, but to construct predictive sets for the censored observed response under label shift.

The closest comparison is with uncensored split conformal Bayes under continuous label shift \citep{Choi2026eiml}.  In that setting, Gaussian predictive tilting shifts the posterior predictive mean and the resulting highest density set is an interval.  In the censored setting, tilting produces a three-part normalizer, the density ratio separates into atom ratios and an interior ratio, and the highest density region may include boundary atoms or reduce to a boundary atom alone.  The distinction between latent label shift and observed mixed-space label shift is also new.  A latent exponential tilt induces tail-averaged atom weights, which need not equal simple marginal atom ratios.  This is the main technical reason why the censored extension is not just a notational variant of the uncensored Gaussian case.  \Cref{tab:comparison} in \Cref{app:comparison} summarizes these differences point by point.

\section{Model: Two-Sided Censored Gaussian with a Bayesian Prediction Head}\label{sec:model}

We now specify the censored Gaussian predictive model and the observed-scale density ratios 
used later for weighted conformal calibration.

\subsection{Latent Gaussian model and censoring}

Let $Y^*$ be a latent continuous response and let $\wt Y$ be the censored observation
\begin{equation}\label{eq:censoring-model-section}
\wt Y=T(Y^*)=
\begin{cases}
L, & Y^*\le L,\\
Y^*, & L<Y^*<U,\\
U, & Y^*\ge U.
\end{cases}
\end{equation}
The Bayesian prediction head places a Gaussian model on top of a fixed representation $\backbone:\cX\to\R^d$:
\begin{equation}\label{eq:latent}
Y_i^*\mid x_i,\theta
\sim
\cN\!\left(\theta^\top\backbone(x_i),\sigeps^2\right),
\qquad
\theta\sim\cN(0,\tau^2 I_d).
\end{equation}
Only the final linear head is inferred.  The representation may come from a pretrained encoder 
or any fixed feature extractor.  
We write \(\sigeps^2\) for the latent Gaussian noise variance and \(\bhsigeps^2\) for its fitted value.

Because observations are censored at both \(L\) and \(U\), the likelihood is a
two-sided Tobit likelihood, also known as a two-limit Tobit likelihood
\citep{TobinJ1958econ, RosettRN1975econ}.
Let \(\Phi\) and \(\phist\) denote the cumulative distribution function (CDF) and
density of the standard normal distribution, respectively.  For a single
observation,
\begin{equation}\label{eq:tobit}
p(\wt y_i\mid x_i,\theta)=
\begin{cases}
\Phi\!\left(\dfrac{L-\theta^\top\backbone(x_i)}{\sigeps}\right),
& \wt y_i=L,\\[1.2ex]
\dfrac{1}{\sigeps}
\phist\!\left(\dfrac{\wt y_i-\theta^\top\backbone(x_i)}{\sigeps}\right),
& L<\wt y_i<U,\\[1.2ex]
1-\Phi\!\left(\dfrac{U-\theta^\top\backbone(x_i)}{\sigeps}\right),
& \wt y_i=U.
\end{cases}
\end{equation}
An interior observation contributes a Gaussian density, while a censored observation contributes a Gaussian tail probability.  
Then we write \(\ell(\theta;\sigeps)\) for the observed-data log-likelihood of the two-sided Tobit model, 
viewed as a function of \(\theta\) conditional on the noise scale.  
\begin{align}
\ell(\theta;\sigeps)
& =
\sum_{i:\wt Y_i=L}
\log \Phi\!\left(\frac{L-\theta^\top\backbone(x_i)}{\sigeps}\right)
+
\sum_{i:L<\wt Y_i<U}
\left[
-\log\sigeps
+
\log \phist\!\left(\frac{\wt Y_i-\theta^\top\backbone(x_i)}{\sigeps}\right)
\right]  \notag \\
& \quad+
\sum_{i:\wt Y_i=U}
\log \left\{
1-\Phi\!\left(\frac{U-\theta^\top\backbone(x_i)}{\sigeps}\right)
\right\}.
\label{eq:loglik}
\end{align}
With the Gaussian prior $\theta\sim\cN(0,\tau^2 I_d)$, we approximate the posterior of $\theta$ by a Laplace approximation at the MAP,
\begin{equation}\label{eq:lappost}
\pi(\theta\mid\dtrain)
\approx
\cN(\bhtheta,\bhSigma_\theta),
\end{equation}
where
\begin{equation}\label{eq:hessian}
\bhtheta
=
\arg\max_\theta
\left\{
\ell(\theta;\bhsigeps)-\frac{1}{2\tau^2}\norm{\theta}^2
\right\},
\qquad
\bhSigma_\theta
=
\left(
\frac{1}{\tau^2}I_d-\nabla_\theta^2\ell(\bhtheta;\bhsigeps)
\right)^{-1}.
\end{equation}
Details on the gradient, Mills ratio derivative, and Hessian are given in \Cref{app:hessian}.

\subsection{Mixed predictive distribution}

The Laplace approximation gives a Gaussian predictive distribution for the latent response.  For a new input $x$,
\begin{equation}\label{eq:predmeanvar}
Y^*\mid x,\dtrain
\sim
\cN\!\left(\bhmu(x),\bhsigma^2(x)\right),
\qquad
\bhmu(x)=\bhtheta^\top\backbone(x),
\quad
\bhsigma^2(x)=\bhsigeps^2+\backbone(x)^\top\bhSigma_\theta\backbone(x).
\end{equation}
After censoring, the predictive distribution of $\wt Y$ given $X=x$ is mixed:
\begin{equation}\label{eq:mixedpred}
dP_s(\wt y\mid x,\dtrain)
=
\pi_L(x)\delta_L(d\wt y)
+
\fs(\wt y\mid x)\mathbf 1\{L<\wt y<U\}\,d\wt y
+
\pi_U(x)\delta_U(d\wt y).
\end{equation}
The conditional atom masses are
\begin{equation}\label{eq:atommasses}
\pi_L(x)=
\Phi\!\left(\frac{L-\bhmu(x)}{\bhsigma(x)}\right),
\qquad
\pi_U(x)=
1-\Phi\!\left(\frac{U-\bhmu(x)}{\bhsigma(x)}\right),
\end{equation}
and the interior density is
\begin{equation}\label{eq:interior-density-model}
\fs(\wt y\mid x)=
\frac{1}{\bhsigma(x)}
\phist\!\left(\frac{\wt y-\bhmu(x)}{\bhsigma(x)}\right),
\qquad L<\wt y<U.
\end{equation}
Here $\fs$ is the ordinary Gaussian density evaluated on the interior interval.  It is not renormalized over $(L,U)$, since
\[
\pi_L(x)+\int_L^U \fs(y\mid x)\,dy+\pi_U(x)=1.
\]

It is convenient to view \eqref{eq:mixedpred} as a density with respect to a single mixed dominating measure
\begin{equation}\label{eq:mixed-measure-def}
\mu=\delta_L+\mathrm{Leb}_{(L,U)}+\delta_U,
\end{equation}
where $\delta_L,\delta_U$ are unit point masses at the censoring limits and $\mathrm{Leb}_{(L,U)}$ is 
Lebesgue measure on the open interior.  With respect to $\mu$, the conditional mixed predictive density is
\begin{equation}\label{eq:mixed-density-mu}
\ps(\wt y\mid x,\dtrain)
=
\pi_L(x)\,\mathbf 1\{\wt y=L\}
+
\fs(\wt y\mid x)\,\mathbf 1\{L<\wt y<U\}
+
\pi_U(x)\,\mathbf 1\{\wt y=U\},
\end{equation}
so that at the atoms the density value equals the atom mass and on the interior it equals the Gaussian density.  The tilted predictive distribution introduced below is likewise a density with respect to the same $\mu$.  These mixed densities are exactly what the negative-log-density scores in \Cref{sec:routes} evaluate, and the prediction-set level sets in \Cref{sec:hpd} are taken with respect to $\mu$.

\subsection{Marginal observed label ratio}

The marginal observed label ratio is given by
\begin{equation}\label{eq:observed-label-ratio}
w(\wt y)=
\frac{dP^t_{\wt Y}}{dP^s_{\wt Y}}(\wt y).
\end{equation}
The marginal observed label distribution has the same three component structure as the predictive distribution 
because censoring forces $\wt Y$ to live on $\{L\}\cup(L,U)\cup\{U\}$.  
The components are different, however, because they are averaged over the input distribution.  
For each domain $j\in\{s,t\}$, we write
\begin{equation}\label{eq:mixed-marginal-source-target}
dP^j_{\wt Y}(\wt y)
=
\rho_{j,L}\delta_L(d\wt y)
+
g_j(\wt y)\mathbf 1\{L<\wt y<U\}\,d\wt y
+
\rho_{j,U}\delta_U(d\wt y).
\end{equation}
For example, under the source distribution,
\[
\rho_{s,L}=P_s(\wt Y=L),
\qquad
\rho_{s,U}=P_s(\wt Y=U),
\qquad
g_s(y)=\int \fs(y\mid x)\,dP_s^X(x),
\quad L<y<U.
\]
Thus $\rho_{s,L}$ is not the same object as $\pi_L(x)$.  
The former is a marginal censoring probability, while the latter is a predictive atom mass at a particular input $x$.

Taking the Radon Nikodym derivative of \eqref{eq:mixed-marginal-source-target} gives the three component importance weight
\begin{equation}\label{eq:three-component-ratio-model}
w(\wt y)=
\begin{cases}
\dfrac{\rho_{t,L}}{\rho_{s,L}}, & \wt y=L,\\[1.2ex]
\dfrac{g_t(\wt y)}{g_s(\wt y)}, & L<\wt y<U,\\[1.2ex]
\dfrac{\rho_{t,U}}{\rho_{s,U}}, & \wt y=U.
\end{cases}
\end{equation}
The weight $w(\wt y)$ is not a single ordinary density ratio on $\mathbb R$.  
It is a mixed-space ratio.  At the left boundary it is a ratio of censoring probabilities, 
on the interior it is an ordinary density ratio, and at the right boundary it is again a ratio of censoring probabilities.  
This is the central difference from the uncensored Gaussian setting.

\subsection{Predictive tilting by the observed label ratio}

The same marginal weight $w(\wt y)$ plays two distinct roles.  First, it weights calibration examples in weighted conformal calibration.  This use is justified by the observed-space importance identity
\begin{equation}\label{eq:obsimportance}
\E_s\{w(\wt Y)h(X,\wt Y)\}
=
\E_t\{h(X,\wt Y)\}
\qquad
\text{for all measurable }h.
\end{equation}
Second, under the predictive tilting approximation used in split conformal Bayes \citep{Choi2026eiml}, the same observed label ratio is used to modify the source predictive distribution and form a target-aligned score.  This second use is stronger than calibration weighting alone.  It assumes that the fitted source predictive distribution can be transported to the target environment by tilting its observed label coordinate while keeping the input $x$ fixed.  Under this approximation,
\begin{equation}\label{eq:predictive-tilt-general}
\pt(\wt y\mid x,\dtrain)
=
\frac{\ps(\wt y\mid x,\dtrain)w(\wt y)}{Z(x)},
\end{equation}
where the normalizer is the mixed integral
\begin{equation}\label{eq:Zclosed}
Z(x)
=
\pi_L(x)w(L)
+
\int_L^U \fs(y\mid x)w(y)\,dy
+
\pi_U(x)w(U).
\end{equation}
Consequently,
\begin{equation}\label{eq:predictive-ratio}
\frac{\pt(\wt y\mid x,\dtrain)}{\ps(\wt y\mid x,\dtrain)}
=
\frac{w(\wt y)}{Z(x)}.
\end{equation}
The ratio in \eqref{eq:predictive-ratio} is a predictive ratio at a fixed input $x$.  It is not the calibration weight.  The calibration weight remains the marginal observed label ratio $w(\wt y)$ in \eqref{eq:three-component-ratio-model}.

In practice we use the parametric exponential family inherited from Gaussian split conformal Bayes to model the interior ratio.  For $L<\wt y<U$,
\begin{equation}\label{eq:winterior}
w(\wt y)=\frac{\exp(\beta\wt y)}{Z_w}.
\end{equation}
The two atom weights are scalar marginal ratios,
\begin{equation}\label{eq:scalaratoms}
w(L)=\frac{\rho_{t,L}}{\rho_{s,L}},
\qquad
w(U)=\frac{\rho_{t,U}}{\rho_{s,U}}.
\end{equation}
This gives the practical tilted predictive distribution
\begin{equation}\label{eq:tiltedpred}
dP_t(\wt y\mid x,\dtrain)
=
\pi_L^t(x)\delta_L(d\wt y)
+
\ft(\wt y\mid x)\mathbf 1\{L<\wt y<U\}\,d\wt y
+
\pi_U^t(x)\delta_U(d\wt y),
\end{equation}
with
\begin{equation}\label{eq:tiltedparts}
\pi_L^t(x)=\frac{\pi_L(x)w(L)}{Z(x)},
\qquad
\pi_U^t(x)=\frac{\pi_U(x)w(U)}{Z(x)},
\end{equation}
and
\begin{equation}\label{eq:tilted-density-model}
\ft(\wt y\mid x)=
\frac{\fs(\wt y\mid x)\exp(\beta\wt y)/Z_w}{Z(x)},
\qquad L<\wt y<U.
\end{equation}
The interior integral in \eqref{eq:Zclosed} is available in closed form:
\begin{equation}\label{eq:interior-mgf-model}
\int_L^U \fs(y\mid x)\frac{\exp(\beta y)}{Z_w}\,dy
=
\frac{\exp\!\left(\beta\bhmu(x)+\frac{1}{2}\beta^2\bhsigma^2(x)\right)}{Z_w}
\{\Phi(b(x))-\Phi(a(x))\},
\end{equation}
where
\begin{equation}\label{eq:ab}
a(x)=\frac{L-\bhmu(x)-\beta\bhsigma^2(x)}{\bhsigma(x)},
\qquad
b(x)=\frac{U-\bhmu(x)-\beta\bhsigma^2(x)}{\bhsigma(x)}.
\end{equation}
On the interior, the tilted density is proportional to a Gaussian density with shifted mean
\begin{equation}\label{eq:tilted-mean-model}
\mu_t(x)=\bhmu(x)+\beta\bhsigma^2(x).
\end{equation}
The atom masses are adjusted separately through $w(L)$ and $w(U)$.

Combining \eqref{eq:interior-mgf-model} with the two atom terms in
\eqref{eq:Zclosed} gives the practical closed-form normalizer
\begin{equation}\label{eq:Zclosed-practical}
Z(x)
=
\pi_L(x)w(L)
+
\frac{\exp\!\left(\beta\bhmu(x)+\frac12\beta^2\bhsigma^2(x)\right)}{Z_w}
\{\Phi(b(x))-\Phi(a(x))\}
+
\pi_U(x)w(U).
\end{equation}
The tilted atom masses and interior density are given by \eqref{eq:tiltedparts}
and \eqref{eq:tilted-density-model}.

\subsection{From latent tilt to observed atom weights}

The previous subsection defines the method in observed space.  
The importance weight is the marginal observed label ratio $w(\wt y)=\frac{dP^t_{\wt Y}}{dP^s_{\wt Y}}(\wt y)$, and the tilted predictive distribution is obtained by multiplying the source predictive distribution by this observed ratio and normalizing.  This subsection explains how such a mixed observed weight arises when the shift is modeled first on the latent scale.

Suppose that the latent response satisfies a label shift model,
\begin{equation}\label{eq:labelshift}
P_s(X\mid Y^*)=P_t(X\mid Y^*),
\qquad
P_s(Y^*)\neq P_t(Y^*),
\end{equation}
with an exponential latent ratio
\begin{equation}\label{eq:exptilt}
w^*(y^*)=\frac{dP_t^*}{dP_s^*}(y^*)=\frac{\exp(\beta y^*)}{Z_w}.
\end{equation}
On the interior, censoring is invertible, so the latent ratio pushes forward to the pointwise interior ratio in \eqref{eq:winterior}.  At the atoms, however, censoring collapses an entire latent tail.  The observed atom ratios are therefore tail averages:
\begin{equation}\label{eq:wLU}
w(L)=
\frac{\int_{-\infty}^{L}\exp(\beta y^*)\,dP_s^*(y^*)}{Z_w P_s(Y^*\le L)},
\qquad
w(U)=
\frac{\int_{U}^{\infty}\exp(\beta y^*)\,dP_s^*(y^*)}{Z_w P_s(Y^*\ge U)}.
\end{equation}
For a Gaussian source marginal $Y^*\sim\cN(\mu_p,\sigma_p^2)$, this becomes
\begin{equation}\label{eq:wLU_closed}
w(L)=
\frac{\Phi\!\left((L-\mu_p-\beta\sigma_p^2)/\sigma_p\right)}
     {\Phi\!\left((L-\mu_p)/\sigma_p\right)},
\qquad
w(U)=
\frac{1-\Phi\!\left((U-\mu_p-\beta\sigma_p^2)/\sigma_p\right)}
     {1-\Phi\!\left((U-\mu_p)/\sigma_p\right)}.
\end{equation}
A negative tilt increases the left atom ratio, while a positive tilt increases the right atom ratio.

There is one additional distinction between the exact latent pushforward and the practical observed-space implementation.  If the latent tilt is applied to the predictive distribution at a fixed input $x$, the exact atom weights are generally input-dependent:
\begin{equation}\label{eq:wxatoms}
w_x(L)=
\frac{\int_{-\infty}^{L}p_s(y^*\mid x,\dtrain)w^*(y^*)\,dy^*}{\pi_L(x)},
\qquad
w_x(U)=
\frac{\int_{U}^{\infty}p_s(y^*\mid x,\dtrain)w^*(y^*)\,dy^*}{\pi_U(x)}.
\end{equation}
The practical method instead uses the scalar marginal atom ratios in \eqref{eq:scalaratoms}.  Thus the latent model is useful as a source of the observed mixed-space weight, but the algorithm itself is defined by the observed marginal ratio in \eqref{eq:three-component-ratio-model} and the tilted predictive distribution in \eqref{eq:predictive-tilt-general}.  
The full conditional pushforward calculation is shown in \Cref{app:latent-pushforward}.

\begin{remark}[Latent shift becomes a mixed observed-space weight]\label{rmk:latent_observed_shift}
A latent exponential tilt does not push forward to a single ordinary density ratio on the observed censored scale.  On the interior, the censoring map is invertible and the observed weight is the pointwise density ratio in \eqref{eq:winterior}.  At each boundary, censoring collapses an entire latent tail, so the latent tilt induces tail-averaged atom weights.  These conditional tail averages are generally input-dependent, as in \eqref{eq:wxatoms}.  Weighted conformal calibration, however, uses the marginal observed-space ratio in \eqref{eq:three-component-ratio-model}.  Thus the conformal weight is the mixed object with a left atom ratio, an interior density ratio, and a right atom ratio, rather than the latent ratio itself.
\end{remark}

\begin{remark}[Marginal atom-weight approximation]\label{rmk:marginalapprox}
The practical tilted predictive distribution replaces the input-dependent atom weights \(w_x(L),w_x(U)\) in \eqref{eq:wxatoms} by scalar marginal atom ratios in \eqref{eq:scalaratoms}.  These scalar ratios are the correct atom components of the marginal observed-space importance weight.  As substitutes inside the predictive tilt at a fixed input \(x\), they are exact only when the tail-averaged atom weights do not depend on \(x\), or when additional model structure makes the marginal and conditional atom corrections agree.  Otherwise this replacement is an approximation, and it is the main censoring-specific distinction between the exact latent-pushforward construction and the practical observed-space implementation.
\end{remark}

The tilted predictive distribution in \eqref{eq:tiltedpred} defines the conformal score, 
while the marginal observed-label ratio in \eqref{eq:three-component-ratio-model} defines the calibration weight.  
The next section combines these two ingredients to define the UT, WS, and WT calibration configurations.

\section{Calibration and Prediction Sets on the Mixed Space}\label{sec:method}\label{sec:routes}

This section turns the mixed-space ratio from \Cref{sec:model} into a practical split conformal Bayes procedure. The method has four main ingredients: choose a score--weight pair, compute a weighted conformal threshold, convert the WT threshold into a mixed-space HDR set, and estimate the mixed density ratio used by the practical algorithm. The recommended configuration is WT, which uses the observed label ratio for calibration and the tilted mixed predictive distribution for set geometry.  \Cref{fig:method-schematic} gives the overall pipeline, and \Cref{fig:ut-ws-wt} summarizes the three calibration configurations.

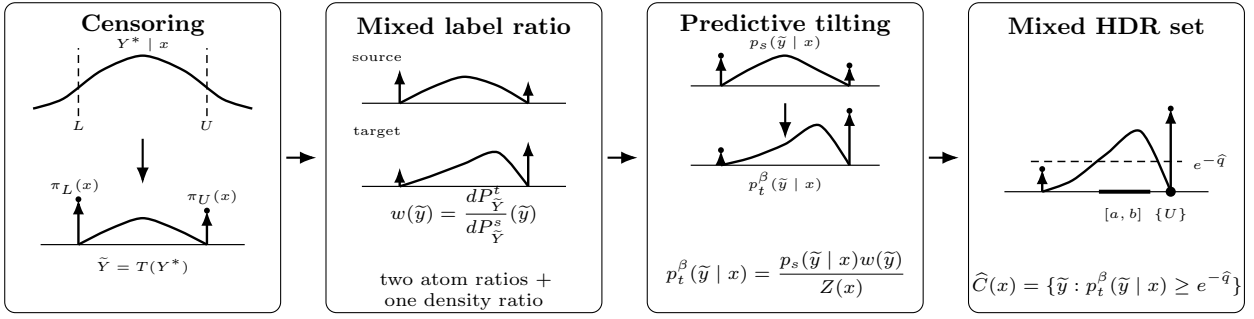
\begin{figure*}[t]
\centering
\begin{tikzpicture}[
  x=1cm,y=1cm,
  >=Latex,
  panel/.style={draw, rounded corners, minimum width=3.65cm, minimum height=4.1cm},
  title/.style={font=\bfseries\small, align=center},
  note/.style={font=\scriptsize, align=center},
  tiny/.style={font=\tiny, align=center},
  arr/.style={-{Latex[length=2.2mm]}, thick},
  atom/.style={-{Latex[length=2mm]}, line width=0.9pt},
  base/.style={line width=0.45pt},
  curve/.style={line width=0.9pt},
  th/.style={densely dashed, line width=0.55pt}
]

\begin{scope}[shift={(0,0)}]
\node[panel] (P1) at (0,0) {};
\node[title] at (0,1.75) {Censoring};

\draw[curve]
  plot[smooth] coordinates {(-1.45,0.65) (-1.05,0.78) (-0.55,1.15) (0,1.35) (0.55,1.15) (1.05,0.78) (1.45,0.65)};
\draw[th] (-0.85,0.55) -- (-0.85,1.45);
\draw[th] (0.85,0.55) -- (0.85,1.45);
\node[tiny] at (-0.85,0.42) {$L$};
\node[tiny] at (0.85,0.42) {$U$};
\node[tiny] at (0,1.50) {$Y^*\mid x$};

\draw[arr] (0,0.25) -- (0,-0.35);

\draw[base] (-1.35,-1.15) -- (1.35,-1.15);
\draw[curve]
  plot[smooth] coordinates {(-0.85,-1.15) (-0.45,-0.95) (0,-0.80) (0.45,-0.95) (0.85,-1.15)};
\draw[atom] (-0.85,-1.15) -- (-0.85,-0.55);
\draw[atom] (0.85,-1.15) -- (0.85,-0.70);
\fill (-0.85,-0.55) circle (1.1pt);
\fill (0.85,-0.70) circle (1.1pt);
\node[tiny] at (-0.90,-0.38) {$\pi_L(x)$};
\node[tiny] at (0.92,-0.52) {$\pi_U(x)$};
\node[tiny] at (0,-1.43) {$\widetilde Y=T(Y^*)$};
\end{scope}

\begin{scope}[shift={(4.25,0)}]
\node[panel] (P2) at (0,0) {};
\node[title] at (0,1.75) {Mixed label ratio};

\draw[base] (-1.35,0.72) -- (1.35,0.72);
\draw[curve]
  plot[smooth] coordinates {(-0.85,0.72) (-0.45,0.92) (0,1.07) (0.45,0.95) (0.85,0.72)};
\draw[atom] (-0.85,0.72) -- (-0.85,1.17);
\draw[atom] (0.85,0.72) -- (0.85,1.02);
\node[tiny] at (-1.15,1.32) {source};

\draw[base] (-1.35,-0.38) -- (1.35,-0.38);
\draw[curve]
  plot[smooth] coordinates {(-0.85,-0.38) (-0.45,-0.25) (0,-0.08) (0.45,0.07) (0.85,-0.38)};
\draw[atom] (-0.85,-0.38) -- (-0.85,-0.13);
\draw[atom] (0.85,-0.38) -- (0.85,0.22);
\node[tiny] at (-1.15,0.34) {target};

\node[note, text width=3.05cm] at (0,-1)
{\[
w(\widetilde y)=
\frac{dP^t_{\widetilde Y}}{dP^s_{\widetilde Y}}(\widetilde y)
\]
{\scriptsize two atom ratios + one density ratio}};
\end{scope}

\begin{scope}[shift={(8.50,0)}]
\node[panel] (P3) at (0,0) {};
\node[title] at (0,1.75) {Predictive tilting};

\draw[base] (-1.25,0.95) -- (1.25,0.95);
\draw[curve]
  plot[smooth] coordinates {(-0.85,0.95) (-0.45,1.18) (0.00,1.35) (0.45,1.15) (0.85,0.95)};
\draw[atom] (-0.85,0.95) -- (-0.85,1.35);
\draw[atom] (0.85,0.95) -- (0.85,1.22);
\fill (-0.85,1.35) circle (1.1pt);
\fill (0.85,1.22) circle (1.1pt);
\node[tiny] at (0,1.52) {$p_s(\widetilde y\mid x)$};

\draw[arr] (0,0.72) -- (0,0.25);

\draw[base] (-1.25,-0.10) -- (1.25,-0.10);
\draw[curve]
  plot[smooth] coordinates {(-0.85,-0.10) (-0.45,0.00) (0.00,0.18) (0.45,0.43) (0.85,-0.10)};
\draw[atom] (-0.85,-0.10) -- (-0.85,0.10);
\draw[atom] (0.85,-0.10) -- (0.85,0.62);
\fill (-0.85,0.10) circle (1.1pt);
\fill (0.85,0.62) circle (1.1pt);
\node[tiny] at (0,-0.34) {$p_t^\beta(\widetilde y\mid x)$};

\node[note, text width=3.15cm] at (0,-1.35)
{\[
p_t^\beta(\widetilde y\mid x)
=
\frac{p_s(\widetilde y\mid x)w(\widetilde y)}
     {Z(x)}
\]};
\end{scope}

\begin{scope}[shift={(12.75,0)}]
\node[panel] (P4) at (0,0) {};
\node[title] at (0,1.75) {Mixed HDR set};

\draw[base] (-1.35,-0.45) -- (1.35,-0.45);
\draw[curve]
  plot[smooth] coordinates {(-0.85,-0.45) (-0.45,-0.30) (0,0.05) (0.45,0.35) (0.85,-0.45)};
\draw[atom] (-0.85,-0.45) -- (-0.85,-0.15);
\draw[atom] (0.85,-0.45) -- (0.85,0.65);
\fill (-0.85,-0.15) circle (1.1pt);
\fill (0.85,0.65) circle (1.1pt);

\draw[th] (-1.00,-0.05) -- (1.00,-0.05);
\node[tiny, anchor=west] at (1.02,-0.05) {$e^{-\widehat q}$};

\draw[line width=1.5pt] (-0.10,-0.45) -- (0.58,-0.45);
\fill (0.85,-0.45) circle (1.9pt);
\node[tiny] at (0.22,-0.75) {$[a,b]$};
\node[tiny] at (0.85,-0.75) {$\{U\}$};

\node[note, text width=3.05cm] at (0,-1.48)
{\[
\hspace*{-.1in}
\widehat C(x)=
\{\widetilde y:
p_t^\beta(\widetilde y\mid x)\ge e^{-\widehat q}\}
\]};
\end{scope}

\draw[arr] (1.90,0) -- (2.30,0);
\draw[arr] (6.15,0) -- (6.55,0);
\draw[arr] (10.40,0) -- (10.80,0);

\end{tikzpicture}

\caption{
Schematic of split conformal Bayes on the censored mixed space.
Censoring maps a latent Gaussian response to an observed response with two boundary atoms and one continuous interior density.
The observed label ratio
\(w(\widetilde y)=dP^t_{\widetilde Y}/dP^s_{\widetilde Y}(\widetilde y)\)
is therefore a mixed-space ratio, consisting of two atom-mass ratios and one interior density ratio.
This ratio tilts the source predictive distribution, and weighted conformal calibration selects a density threshold.
The resulting prediction set is a mixed highest-density region that may contain boundary atoms as well as an interior interval.
All predictive densities $p_s,p_t^\beta$ are conditioned on the input $x$ and the training data $\dtrain$; this conditioning is suppressed in the panels for readability.
}
\label{fig:method-schematic}
\end{figure*}

\subsection{Calibration configurations: UT, WS, and WT}
\label{subsec:ut_ws_wt}

The calibration set is drawn from $\ps$, the test point from $\pt$. Label shift requires two distinct corrections: the importance weight $w(\wt{y})$ corrects the calibration \emph{measure}, while the tilted Bayesian score corrects the \emph{geometry} of the candidate set. Censoring does not alter this logic, but it does make the score a mixed-space object. We use the method labels themselves as superscripts and define each configuration by a score--weight pair $(s^r,w^r)$, $r\in\{\mathrm{UT},\mathrm{WS},\mathrm{WT}\}$:
\begin{equation}\label{eq:scores}
\begin{aligned}
s_i^{\mathrm{UT}}&=-\log\pt(\wt{y}_i\mid x_i,\dtrain),
& w_i^{\mathrm{UT}}&=1,\\
s_i^{\mathrm{WS}}&=-\log\ps(\wt{y}_i\mid x_i,\dtrain),
& w_i^{\mathrm{WS}}&=w(\wt{y}_i),\\
s_i^{\mathrm{WT}}&=-\log\pt(\wt{y}_i\mid x_i,\dtrain)
=s_i^{\mathrm{WS}}-\log w(\wt{y}_i)+\log Z(x_i),
& w_i^{\mathrm{WT}}&=w(\wt{y}_i).
\end{aligned}
\end{equation}
Each quantity is evaluated on the appropriate mixed component: $w(L),w(U)$ at atoms, $e^{\beta\wt{y}}/Z_w$ on the interior, and closed-form $Z(x_i)$ from \eqref{eq:Zclosed-practical}. Because $s^{\mathrm{WT}}$ is a negative log-density under the \emph{target} mixed predictive, its fixed-threshold level sets are target mixed-space HDRs. The practical calibration-only thresholds are
\begin{equation}\label{eq:routes}
\bq_{\mathrm{UT}}=\mathrm{Quantile}_{1-\alpha}\big(\{s_i^{\mathrm{UT}}\}\big),\qquad
\bq_r=\inf\Big\{q:\tfrac{\sum_i w_i^r\ind\{s_i^r\le q\}}{\sum_j w_j^r}\ge 1-\alpha\Big\},\ \ r\in\{\mathrm{WS},\mathrm{WT}\}.
\end{equation}
UT uses tilted geometry but no measure correction; WS uses measure correction but source geometry; WT uses both and is the recommended procedure.  This notation avoids the artificial mapping from earlier route labels to the three method names.

\medskip\noindent\textbf{Oracle exact rule.}
The exact finite-sample construction adds the candidate's own test weight. For $r\in\{\mathrm{WS},\mathrm{WT}\}$ and candidate $\wt{y}$, let $w^r(\wt y)=w(\wt y)$ and evaluate the candidate score $s^r(\wt y,x_{\mathrm{test}})$ with the corresponding source or tilted predictive. Then
\begin{equation}\label{eq:exact}
\bq_r(\wt{y})=\inf\Big\{q:\tfrac{\sum_{i}w_i^r\ind\{s_i^r\le q\}+w^r(\wt{y})}{\sum_i w_i^r+w^r(\wt{y})}\ge1-\alpha\Big\},\qquad
\wt{y}\in\cC^r_{\mathrm{exact}}\iff s^r(\wt{y},x_{\mathrm{test}})\le\bq_r(\wt{y}).
\end{equation}
This candidate-dependent rule has exact finite-sample validity under known weights. The fixed-threshold rule \eqref{eq:routes} approximates it up to $\sup_{\wt{y}}w(\wt{y})/\sum_i w(\wt{y}_i)=O(1/n)$ under bounded weights. Conceptually: the exact WT set is not literally a fixed-level HDR at finite $n$, whereas the practical and population WT rules are.

\begin{table}[t]
\centering
\caption{Calibration configurations in the two-sided censored setting. ``Oracle valid'' refers to the candidate-weighted rule \eqref{eq:exact}, not the calibration-only approximation.}\label{tab:routes}
\smallskip
\begin{tabular}{@{}lcccl@{}}
\toprule
Method & Score & Calibration weight & Oracle valid & Geometry \\
\midrule
UT & $s^{\mathrm{UT}}=-\log\pt$ & $1$ & No & target HDR, miscalibrated \\
WS & $s^{\mathrm{WS}}=-\log\ps$ & $w(\wt y)$ & Yes & source HDR \\
WT & $s^{\mathrm{WT}}=-\log\pt$ & $w(\wt y)$ & Yes & target HDR \\
\bottomrule
\end{tabular}
\end{table}

\begin{figure*}[t]
\centering
\begin{tikzpicture}[
  >=Latex,
  node distance=0.7cm and 0.9cm,
  titlebox/.style={
    draw,
    rounded corners,
    align=center,
    inner sep=5pt,
    minimum height=0.9cm,
    text width=3.2cm,
    font=\small
  },
  method/.style={
    draw,
    rounded corners,
    align=center,
    inner sep=6pt,
    minimum height=4.25cm,
    text width=4.0cm,
    font=\small
  },
  outputbox/.style={
    draw,
    rounded corners,
    align=center,
    inner sep=5pt,
    minimum height=1.55cm,
    text width=4.0cm,
    font=\scriptsize
  },
  auxarrow/.style={-{Latex[length=2mm]}, semithick}
]

\node[titlebox] (src) {
\textbf{Source predictive}\\[-1mm]
\[
p_s(\widetilde y\mid x)
\]
};

\node[titlebox, right=2.2cm of src] (tilt) {
\textbf{Tilted predictive}\\[-1mm]
\[
p_t^\beta(\widetilde y\mid x)
\]
};

\node[titlebox, right=2.2cm of tilt] (wgt) {
\textbf{Importance weight}\\[-1mm]
\[
w(\widetilde y)
=
\frac{dP^t_{\widetilde Y}}{dP^s_{\widetilde Y}}(\widetilde y)
\]
};

\node[method, below=1.2cm of src] (UT) {
{\bfseries UT: unweighted tilted}\\[1mm]

\textbf{Score}
\[
s^{\mathrm{UT}}(x,\widetilde y)
=
-\log p_t^\beta(\widetilde y\mid x)
\]

\textbf{Calibration weight}
\[
w^{\mathrm{UT}}(\widetilde y)=1
\]

\textbf{Threshold}
\[
\widehat q_{\mathrm{UT}}=Q_{1-\alpha}
\]
{\scriptsize unweighted conformal quantile}

\vspace{2mm}
{\scriptsize tilted score + unweighted calibration}
};

\node[method, below=1.2cm of tilt] (WS) {
{\bfseries WS: weighted source}\\[1mm]

\textbf{Score}
\[
s^{\mathrm{WS}}(x,\widetilde y)
=
-\log p_s(\widetilde y\mid x)
\]

\textbf{Calibration weight}
\[
w^{\mathrm{WS}}(\widetilde y)=w(\widetilde y)
\]

\textbf{Threshold}
\[
\widehat q_{\mathrm{WS}}=Q^{w}_{1-\alpha}
\]
{\scriptsize weighted conformal quantile}

\vspace{2mm}
{\scriptsize source score + importance weighting}
};

\node[method, below=1.2cm of wgt] (WT) {
{\bfseries WT: weighted tilted}\\[1mm]

\textbf{Score}
\[
s^{\mathrm{WT}}(x,\widetilde y)
=
-\log p_t^\beta(\widetilde y\mid x)
\]

\textbf{Calibration weight}
\[
w^{\mathrm{WT}}(\widetilde y)=w(\widetilde y)
\]

\textbf{Threshold}
\[
\widehat q_{\mathrm{WT}}=Q^{w}_{1-\alpha}
\]
{\scriptsize weighted conformal quantile}

\vspace{2mm}
{\scriptsize tilted score + importance weighting}
};

\draw[auxarrow] (tilt.south) -- (UT.north);
\draw[auxarrow] (src.south) -- (WS.north);
\draw[auxarrow] (wgt.south) -- (WS.north);
\draw[auxarrow] (tilt.south) -- (WT.north);
\draw[auxarrow] (wgt.south) -- (WT.north);

\node[outputbox, below=0.75cm of UT] (CUT) {
\textbf{Conformal prediction set}\\[0.5mm]
\(
\widehat C^{\mathrm{UT}}(x)
=
\left\{\widetilde y : p_t^\beta(\widetilde y \,|\, x)\ge e^{-\widehat q_{\mathrm{UT}}} \right\}
\)\\[1mm]
{\scriptsize tilted mixed-HDR geometry}
};

\node[outputbox, below=0.75cm of WS] (CWS) {
\textbf{Conformal prediction set}\\[0.5mm]
\(
\widehat C^{\mathrm{WS}}(x)
=
\left\{\widetilde y : p_s(\widetilde y \,|\, x)\ge e^{-\widehat q_{\mathrm{WS}}} \right\}
\)\\[1mm]
{\scriptsize source mixed-HDR geometry}
};

\node[outputbox, below=0.75cm of WT] (CWT) {
\textbf{Conformal prediction set}\\[0.5mm]
\(
\widehat C^{\mathrm{WT}}(x)
=
\left\{\widetilde y : p_t^\beta(\widetilde y \,|\,  x)\ge e^{-\widehat q_{\mathrm{WT}}} \right\}
\)\\[1mm]
{\scriptsize tilted mixed-HDR geometry}
};

\draw[auxarrow] (UT.south) -- (CUT.north);
\draw[auxarrow] (WS.south) -- (CWS.north);
\draw[auxarrow] (WT.south) -- (CWT.north);

\end{tikzpicture}

\caption{
Visual summary of the three calibration configurations in split conformal Bayes.
UT uses the tilted predictive score without importance weighting.
WS uses the source predictive score with the mixed-space importance weight
\(w(\widetilde y)=dP^t_{\widetilde Y}/dP^s_{\widetilde Y}(\widetilde y)\).
WT uses both the tilted predictive score and the importance weight.
All three outputs are mixed-HDR-type prediction sets for their corresponding predictive densities: UT and WT use tilted geometry, while WS uses source geometry.  Among the three, WT combines target-adapted geometry with weighted calibration.
All predictive densities are conditioned on the input $x$ and the training data $\dtrain$; this conditioning is suppressed in the panels for readability.
}
\label{fig:ut-ws-wt}
\end{figure*}
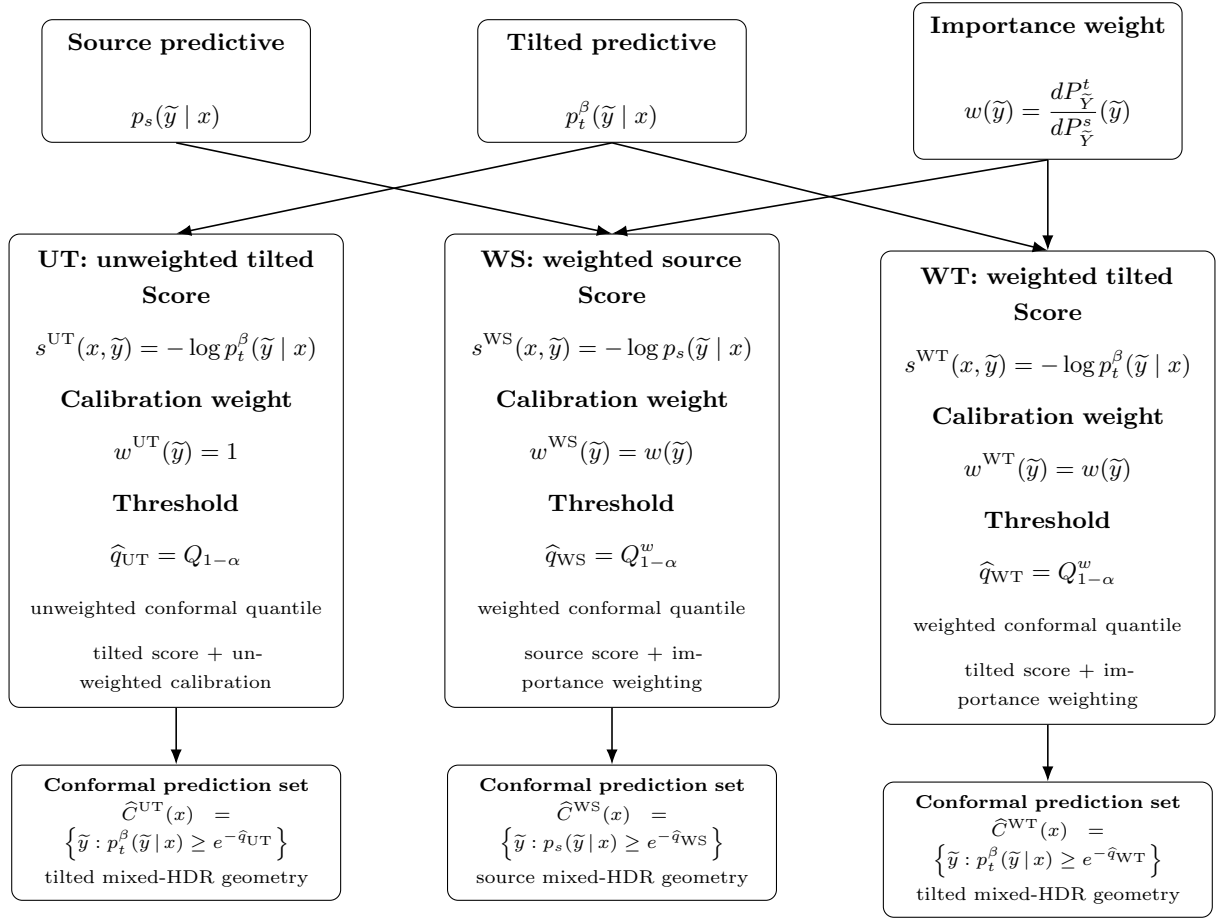

\subsection{Mixed-space HDR prediction set}\label{sec:hpd}

The WT prediction set is a density level set of the tilted mixed predictive distribution.  Let
\[
\lambda^*=e^{-\bq_{\mathrm{WT}}}.
\]
Then
\begin{equation}\label{eq:WT-HDR}
\cC^{\mathrm{WT}}(x)
=
\{\wt{y}:\pt(\wt{y}\mid x,\dtrain)\ge \lambda^*\}.
\end{equation}
Because \(\pt\) has atoms at \(L,U\) and a Lebesgue density on \((L,U)\), this level set is taken with respect to the mixed measure \(\mu\) of \eqref{eq:mixed-measure-def}; equivalently, for a set \(A\),
\begin{equation}\label{eq:mixedmeasure}
\mu(A)=\ind[L\in A]+\mathrm{Leb}(A\cap(L,U))+\ind[U\in A].
\end{equation}
With respect to \(\mu\), the atom masses and the interior density are all compared to the same threshold.  Thus the boundary decisions are simply
\begin{equation}\label{eq:atominclusion}
L\in\cC^{\mathrm{WT}}(x)
\iff
\piLt(x)\ge\lambda^*,
\qquad
U\in\cC^{\mathrm{WT}}(x)
\iff
\piUt(x)\ge\lambda^*.
\end{equation}

It remains to describe the interior part.  Define
\[
I(x)=\{\wt{y}\in(L,U):\ft(\wt{y}\mid x)\ge\lambda^*\}.
\]
The tilted interior density is Gaussian-shaped and unimodal on \((L,U)\), with shifted center
\(\mut(x)=\bhmu(x)+\beta\bhsigma^2(x)\).  Therefore \(I(x)\) is either empty or a single interval.  Solving the Gaussian density inequality gives the half-width and clipped endpoints
\begin{equation}\label{eq:interval}
\begin{aligned}
r(x)&=
\bhsigma(x)
\sqrt{
2\Big[
\beta\bhmu(x)
+\tfrac12\beta^2\bhsigma^2(x)
-\log\!\big(Z(x)Z_w\bhsigma(x)\sqrt{2\pi}\big)
-\log\lambda^*
\Big]},\\
y_\ell(x)&=\max\{L,\mut(x)-r(x)\},
\qquad
y_u(x)=\min\{U,\mut(x)+r(x)\}.
\end{aligned}
\end{equation}
The interior set is empty when the radicand in \eqref{eq:interval} is negative.  Equivalently, \(\lambda^*\) exceeds the maximum interior density, which is \(\ft(\mut(x)\mid x)\) if \(\mut(x)\in(L,U)\) and \(\max\{\ft(L^+\mid x),\ft(U^-\mid x)\}\) otherwise.

\begin{proposition}[Structure of the mixed HDR set]\label{prop:structure}
The WT set decomposes as
\[
\cC^{\mathrm{WT}}(x)=A_L(x)\cup I(x)\cup A_U(x),
\]
where
\[
A_L(x)=
\begin{cases}
\{L\}, & \piLt(x)\ge\lambda^*,\\
\emptyset, & \piLt(x)<\lambda^*,
\end{cases}
\qquad
A_U(x)=
\begin{cases}
\{U\}, & \piUt(x)\ge\lambda^*,\\
\emptyset, & \piUt(x)<\lambda^*,
\end{cases}
\]
and \(I(x)\) is either empty or the interval \([y_\ell(x),y_u(x)]\) in \eqref{eq:interval}.  Consequently:
\begin{enumerate}[label={\rm(\alph*)},leftmargin=2.4em,itemsep=2pt]
\item if \(I(x)\neq\emptyset\), the set is an interior interval possibly unioned with one or both boundary atoms;
\item if \(I(x)=\emptyset\), the set is atom-only, two-atom-only, or empty.
\end{enumerate}
\end{proposition}
\noindent Proof in \Cref{app:proof-structure}.  All components are closed form once \(\bhtheta,\bhSigma_\theta,\bbeta,\bq_{\mathrm{WT}}\) are fixed: the interior endpoints come from \eqref{eq:interval}, and the two atom decisions come from \eqref{eq:atominclusion}.  Under increasing censoring the dominant forms are the boundary-extended interval and, eventually, the atom-only set.  Disconnected forms are possible in the abstract mixed-space representation but are non-generic for the tilted Gaussian prediction head used here (\Cref{rmk:disconnected}).

\begin{remark}[Atom-only and disconnected sets]\label{rmk:disconnected}
Two non-interval outputs deserve comment.  First, under strong censoring the interior super-level set can be empty and \(\cC^{\mathrm{WT}}\) can collapse to an atom-only set, such as \(\{U\}\) under positive shift or \(\{L\}\) under negative shift.  This reports that the response is plausibly beyond a detection limit, and it is the dominant high-censoring regime in \Cref{sec:exp3}.  Second, a disconnected set such as \(\{L\}\cup[y_\ell,y_u]\cup\{U\}\) would require both atom masses to exceed \(\lambda^*\) while the adjacent interior densities fall below it.  For a single tilted Gaussian this configuration is non-generic: the interior is unimodal and the monotone tilt \(e^{\beta\wt{y}}\) raises one atom while suppressing the other.  The same mechanism makes the two-atom-only form \(\{L\}\cup\{U\}\) (empty interior, both atoms retained) non-generic, which is why neither form appears in \Cref{tab:exp3-geometry}.  A genuinely disconnected mixed set would therefore be more naturally associated with a multimodal interior predictive, for example a non-Gaussian or mixture head.
\end{remark}

\subsection{Density-ratio estimation and pseudo-label bias}\label{sec:dre}

The practical method requires estimates of the three components of
\(w(\wt y)\).  Censoring makes these components different in difficulty.
The two atom weights are probabilities of observable events, while the
interior density ratio requires modeling how the target labels would be
distributed inside \((L,U)\).

\medskip\noindent\textbf{Atom weights.}
The boundary weights are
\[
  w(L)=\frac{\Prob_t(\wt Y=L)}{\Prob_s(\wt Y=L)},
  \qquad
  w(U)=\frac{\Prob_t(\wt Y=U)}{\Prob_s(\wt Y=U)}.
\]
When target censoring status or aggregate target assay summaries are
available, these weights can be estimated by censoring-fraction ratios.  If
only target inputs are available, the atom weights must be supplied by a
predictive model, treated as sensitivity parameters, or evaluated in an oracle
experiment.  Rare atoms should be smoothed or clipped because fraction ratios
can have high variance.

\medskip\noindent\textbf{Interior ratio.}
On the interior, the log ratio is modeled as
\[
  \log w(\wt y)=\beta \wt y-\log Z_w,
  \qquad L<\wt y<U.
\]
Estimating \(\beta\) from target inputs requires pseudo-labels.  A point
pseudo-label \(\bhmu(X)\) is attractive but discards the within-input
predictive variance.  This is exactly the uncertainty that the conformal Bayes
score is meant to use.  To see the bias, write the latent response as
\(Y=M+\varepsilon\), with \(\Var_s(M)=a\) and \(\Var_s(\varepsilon)=b\).
An exponential tilt shifts the target mean by \(\beta(a+b)\), with predictable
part \(\beta a\) and residual part \(\beta b\).

\begin{lemma}[Variance compression]\label{lem:varcomp}
With \(v_s^2=\Var_s(Y^*)\) and
\(\bar\sigma_s^2=\E_{\ps(x)}[\bhsigma^2(x)]\), point pseudo-labels
\(\wt Y=\bhmu(X)\) satisfy
\[
  \Var_{\ps(x)}(\wt Y)
  =\Var_{\ps(x)}[\bhmu(X)]
  \approx v_s^2-\bar\sigma_s^2,
\]
removing the within-input variance.
\end{lemma}
\noindent
The approximation uses the posterior predictive variance
\(\bhsigma^2(x)=\bhsigeps^2+\backbone(x)^\top\bhSigma_\theta\backbone(x)\) in
place of the latent conditional variance \(\sigeps^2\).  These coincide only
when posterior parameter uncertainty is negligible, as in the oracle regime used
in \Cref{sec:exp}.

\begin{proposition}[Point and source-predictive pseudo-label bias]\label{prop:lrbias}
In the ideal decomposition, tilting gives
\(Y_t\sim\cN(\beta(a+b),a+b)\) and \(M_t\sim\cN(\beta a,a)\).  Point
pseudo-labels \(\wt Y_{\rm point}=M_t\) compress variance from \(a+b\) to
\(a\) and miss the residual shift \(\beta b\).  Source-predictive sampling
\(\wt Y_{\rm srcPS}=M_t+\varepsilon'\), \(\varepsilon'\sim\cN(0,b)\), restores
variance but still misses \(\beta b\).  A log-linear slope fit to such samples
is therefore attenuated to
\[
  \widehat\beta_{\rm srcPS}\approx \beta\,\frac{a}{a+b}
  =\beta R^2_{\rm BRR}.
\]
\end{proposition}
\noindent Proofs of \Cref{lem:varcomp} and \Cref{prop:lrbias} are given in \Cref{app:proof-lrbias}.

\medskip\noindent\textbf{Tilted predictive sampling.}
A direct way to restore the missing residual shift is to sample from the
tilted predictive distribution.  Given an initial value \(\bbeta^{(0)}\), draw
\begin{equation}\label{eq:tiltpredsamp}
\wt y_j^{(s)}
=
\bhmu(x_j)+\bbeta^{(0)}\bhsigma^2(x_j)
+\bhsigma(x_j)\epsilon_j^{(s)},
\qquad
\epsilon_j^{(s)}\sim\cN(0,1),
\quad s=1,\ldots,S.
\end{equation}
Draws outside \((L,U)\) are assigned to the corresponding atom.  Moment
matching,
\[
  \bbeta^{(0)}_{\rm MM}
  =\frac{\bar\mu_t-\bar\mu_s}{\widehat v_\mu^2+\lambda},
\]
with a small ridge \(\lambda\ge0\), is a stable initializer and a strong
baseline in the oracle linear-Gaussian setting.  If \(\bbeta^{(0)}=\beta\), the
samples in \eqref{eq:tiltpredsamp} have conditional mean
\(\bhmu(x)+\beta\bhsigma^2(x)\) and variance \(\bhsigma^2(x)\), matching the
tilted predictive distribution before censoring.  Thus tilted predictive
sampling is the principled default when the predictive variance is trusted,
while moment matching is useful as an initializer and benchmark.

\begin{table}[t]
\centering
\caption{Bias structure and recommended estimators for the mixed censored ratio.}\label{tab:biassummary}
\smallskip
\begin{tabular}{@{}llll@{}}
\toprule
Component & Estimator & Main issue & Recommendation \\
\midrule
\(w(L),w(U)\) & censoring-fraction ratio & rare-atom noise & smooth or clip if rare \\
\(\beta\) & LR-point & variance compression and missing shift & diagnostic only \\
\(\beta\) & source-predictive sampling & restores variance, attenuates tilt & ablation \\
\(\beta\) & tilted predictive sampling & depends on initializer & principled default \\
\(\beta\) & moment matching & predictable shift only & initializer or strong baseline \\
\bottomrule
\end{tabular}
\end{table}

\subsection{Practical algorithm}\label{sec:algorithm}

The practical WT implementation follows the construction above.  It first fits
the censored Bayesian prediction head, then estimates the mixed ratio, calibrates the
WT score with the estimated weights, and finally returns the closed-form mixed
HDR set.

\begin{algorithm}[t]
\caption{Practical WT SCB-C}\label{alg:scbc}
\begin{algorithmic}[1]
\Require Training data \(\dtrain\), calibration data \(\dcal\), target covariates \(\dtarget=\{x_j\}\), test input \(x_{\mathrm{test}}\), bounds \(L<U\), backbone \(\backbone\), level \(1-\alpha\)
\Ensure Prediction set \(\cC^{\mathrm{WT}}(x_{\mathrm{test}})\)
\State \textbf{Fit censored Bayesian prediction head.} Obtain the MAP mean parameter \(\bhtheta\) and fitted noise variance \(\bhsigeps^2\) from the two-sided Tobit likelihood \eqref{eq:tobit}; compute the Laplace covariance \(\bhSigma_\theta\) from \eqref{eq:hessian}; evaluate \(\bhmu(x)\) and the predictive variance \(\bhsigma^2(x)\) using \eqref{eq:predmeanvar}.
\State \textbf{Estimate mixed ratio.} Using target covariates but not target responses, estimate the interior tilt \(\bbeta\) by tilted predictive sampling as in \Cref{sec:dre}.  Form the interior ratio \(\widehat{w}(\wt y)\propto\exp(\bbeta\wt y)\) on \((L,U)\), and compute \(\widehat{w}(L)\) and \(\widehat{w}(U)\) from the induced tilted marginal masses on the observed censored scale.
\State \textbf{Compute WT calibration scores.} For each \((x_i,\wt y_i)\in\dcal\), compute \(\log\widehat{Z}(x_i)\) from \eqref{eq:Zclosed-practical} and
\[
  s_i^{\mathrm{WT}}
  =-\log \ps(\wt y_i\mid x_i)-\log \widehat{w}(\wt y_i)+\log\widehat{Z}(x_i).
\]
\State \textbf{Calibrate threshold.} Compute \(\bq_{\mathrm{WT}}\) by the weighted quantile rule \eqref{eq:routes} using scores \(\{s_i^{\mathrm{WT}}\}\) and weights \(\{\widehat{w}(\wt y_i)\}\).
\State \textbf{Return mixed HDR set.} Set \(\lambda^*=e^{-\bq_{\mathrm{WT}}}\); include atoms using \eqref{eq:atominclusion}; compute the interior interval using \eqref{eq:interval}; return \(\cC^{\mathrm{WT}}(x_{\mathrm{test}})=A_L\cup I\cup A_U\).
\end{algorithmic}
\end{algorithm}
\noindent
After the Bayesian prediction head and ratio estimates are fixed, prediction for a new
input requires only closed-form evaluations of \(\widehat{Z}(x)\), the two atom masses,
and the interval endpoints.  The main computational cost is the one-time
Laplace Hessian for the Bayesian prediction head.

\section{Validity and Efficiency of Mixed-Space Calibration}\label{sec:theory}

This section separates the two main theoretical claims.  Weighted conformal
validity is an oracle statement about the observed mixed-space importance
weight.  The WT score is an efficiency device: it changes the geometry of the
prediction set, but is not required for validity once the weights are correct.

\subsection{Validity hierarchy: UT, WS, and WT}

\begin{proposition}[Correction hierarchy]\label{prop:hierarchy}
Let \((x_i,\wt y_i)_{i=1}^n\sim\ps\) i.i.d.,
\((x_{\mathrm{test}},\wt Y_{\mathrm{test}})\sim\pt\), and assume a known
mixed-space importance weight \(w(\wt y)\) satisfying \eqref{eq:obsimportance}.
\begin{enumerate}[label={\rm(\roman*)},leftmargin=*]
\item \emph{UT is generally miscalibrated:}
\[
\left|
\Prob_{\pt}\{\wt Y_{\mathrm{test}}\in\cC^{\mathrm{UT}}\}-(1-\alpha)
\right|
\le
\dTV\{\mathcal L_{\ps}(s^{\mathrm{WT}}),\mathcal L_{\pt}(s^{\mathrm{WT}})\}
+O(n^{-1/2}),
\]
and it may over- or under-cover.
\item \emph{WS and WT are oracle valid:} the candidate-weighted sets
\eqref{eq:exact} satisfy
\[
\Prob_{\pt}\{\wt Y_{\mathrm{test}}\in\cC^r_{\mathrm{exact}}\}
\ge 1-\alpha,
\qquad r\in\{\mathrm{WS},\mathrm{WT}\}.
\]
This statement does not require the two-sided Tobit likelihood, Laplace
approximation, or tilted score to be correct.
\item \emph{Calibration-only is practical, not exact:} if \(0<w\le W_{\max}\)
and \(n^{-1}\sum_i w(\wt y_i)\) is bounded away from zero, the fixed-threshold
rule differs from \eqref{eq:exact} by \(O_{\Prob}(1/n)\) uniformly over
bounded-weight candidates.
\end{enumerate}
\end{proposition}
\noindent Proof in \Cref{app:proof-hierarchy}.  The single ingredient behind the oracle
statement is the mixed-space importance identity \eqref{eq:obsimportance}.  The
Bayesian prediction head and the tilted score affect the usefulness and size of the set,
not the oracle coverage proof for a fixed score.  We therefore distinguish an
oracle regime with known weights and candidate-weighted calibration from the
plug-in regime used by the practical algorithm.

\subsection{Efficiency of the mixed-HDR geometry}

\begin{proposition}[Mixed-HDR optimality]\label{prop:NP}
Fix a target mixed density \(p_t(\wt y\mid x)\) with respect to \(\mu\). Among
all rules with \(\Prob_{p_t}\{\wt Y\in C(X)\}\ge1-\alpha\), any minimizer of
\(\E[\mu(C(X))]\) is, up to ties, a global density-threshold rule
\[
C^*(x)=\{\wt y:p_t(\wt y\mid x)\ge\lambda^*\}.
\]
Hence the fixed-threshold WT rule is the population mixed-HDR rule; WS forms
source HDRs and is generally suboptimal under shift.
\end{proposition}
\noindent Proof in \Cref{app:proof-NP}.  The message is orthogonality: weights
correct the calibration aggregation, while the tilted score corrects the
geometry of the returned set.  Under censoring, this geometry controls both the
interior interval and the binary atom-inclusion decisions.  The optimality
statement is with respect to the mixed size convention \(\mu\) in
\eqref{eq:mixedmeasure}, which charges one unit for each included atom and
Lebesgue length on the interior.  It is therefore not invariant to rescaling the
response variable; it is an HDR optimality statement for the adopted mixed
measure.

\section{Experiments}\label{sec:exp}

The theoretical development above separates validity from efficiency.  The experiments therefore use a controlled simulation suite to test candidate-weighted validity on the observed mixed space, the approximation introduced by scalar atom weights and calibration-only thresholds, the efficiency gain from the tilted mixed score, and the behavior of density-ratio estimators.

\noindent\textbf{Common setup.} We use a latent Gaussian regression model
\[
  X\sim\cN(0,I_d),\qquad
  Y^*=\theta^\top X+\varepsilon,\qquad
  \varepsilon\sim\cN(0,\sigeps^2),
\]
with $d=5$, $\|\theta\|_2=1$, and $\sigeps=1$.  The target distribution is generated by latent exponential tilting,
\[
  \frac{dP_t}{dP_s}(x,y^*)
  =\exp\{\beta y^* - \tfrac12\beta^2\Var_s(Y^*)\},
\]
which implies $X_t\sim\cN(\beta\theta,I_d)$ and $Y_t^*\mid X_t=x\sim\cN(\theta^\top x+\beta\sigeps^2,\sigeps^2)$.  The observed response is $\wt Y=T(Y^*)$ with two-sided censoring at the source marginal censoring quantiles.  Unless otherwise stated, the censoring level is $10\%$ per side, $n_{\rm cal}=600$, $n_{\rm test}=1000$, and results are averaged over $200$ independent repetitions.  We use the oracle Gaussian predictive model in these diagnostics, so the experiments isolate the conformal layer and density-ratio approximations rather than posterior fitting error.  Mean mixed-space size is measured with respect to the mixed measure $\mu$ of \eqref{eq:mixed-measure-def}: it is the interior Lebesgue length plus one unit for each included atom.

The compared rules are as follows.  Source CP uses the source mixed score and an unweighted quantile.  UT uses the target tilted mixed score and an unweighted quantile.  WS uses the source mixed score with the exact observed-space candidate weight.  WT exact uses the target tilted mixed score with the exact observed-space candidate weight.  WT scalar uses the practical scalar-atom plug-in and calibration-only threshold.  The exact candidate-weighted rules use the true observed joint ratio, including the input-specific atom ratios induced by latent pushforward; set sizes for these rules are approximated by a fine grid over $(L,U)$.

\subsection{Experiment 1: method hierarchy}\label{sec:exp1}

\Cref{tab:exp1-hierarchy} gives the basic method hierarchy.  Under no shift, all methods are near the nominal level.  Under strong positive or negative shift, the unweighted source score undercovers mildly, while UT severely overcovers because its target-aligned score is not paired with target/source measure correction.  WS restores marginal coverage but remains wide.  WT exact restores marginal coverage with much smaller mixed-space sets; for $\beta=1.2$, the mean size drops from $2.894$ for WS to $1.740$ for WT exact, a reduction of about $40\%$.  WT scalar nearly matches the marginal coverage and size of WT exact in this controlled setting.

The component-wise columns should not be read as conditional guarantees.  They show the opposite: pooled mixed-space calibration reallocates coverage toward the censoring atom favored by the target shift.  For $\beta=1.2$, WT exact has right-atom coverage $0.990$ but interior coverage $0.730$ and left-atom coverage $0.064$; for $\beta=-1.2$, the pattern reverses.  This is why the main theory is stated as a marginal mixed-space coverage result, not a component-wise result.

\begin{table}[t]
\centering
\caption{Experiment 1: method hierarchy under latent label shift and two-sided censored observation. Values are averages over 200 repetitions.}
\label{tab:exp1-hierarchy}
\small
\setlength{\tabcolsep}{4pt}
\begin{tabular}{llccccc}
\toprule
$\beta$ & Method & Coverage & Size & Left cov. & Interior cov. & Right cov. \\
\midrule
-1.2 & Source CP & 0.880 & 2.761 & 0.866 & 0.910 & 0.763 \\
-1.2 & UT & 0.992 & 2.826 & 1.000 & 0.979 & 0.592 \\
-1.2 & WS exact & 0.905 & 2.904 & 0.894 & 0.928 & 0.782 \\
-1.2 & WT exact & 0.902 & 1.737 & 0.990 & 0.731 & 0.093 \\
-1.2 & WT scalar & 0.901 & 1.751 & 0.998 & 0.713 & 0.128 \\
\midrule
0.0 & Source CP & 0.901 & 3.045 & 0.779 & 0.930 & 0.787 \\
0.0 & UT & 0.901 & 3.045 & 0.779 & 0.930 & 0.787 \\
0.0 & WS exact & 0.901 & 3.040 & 0.779 & 0.930 & 0.787 \\
0.0 & WT exact & 0.901 & 3.040 & 0.779 & 0.930 & 0.787 \\
0.0 & WT scalar & 0.899 & 3.031 & 0.776 & 0.928 & 0.784 \\
\midrule
1.2 & Source CP & 0.880 & 2.773 & 0.740 & 0.909 & 0.866 \\
1.2 & UT & 0.992 & 2.837 & 0.572 & 0.978 & 1.000 \\
1.2 & WS exact & 0.902 & 2.894 & 0.754 & 0.926 & 0.890 \\
1.2 & WT exact & 0.901 & 1.740 & 0.064 & 0.730 & 0.990 \\
1.2 & WT scalar & 0.901 & 1.755 & 0.108 & 0.714 & 0.998 \\
\bottomrule
\end{tabular}
\end{table}

\subsection{Experiment 2: approximation gap}\label{sec:exp2}

\Cref{tab:exp2-approx} isolates the practical approximations used by WT.  The exact candidate rule uses the exact observed joint ratio.  The calibration-only rule keeps the exact joint ratio on calibration points but omits the candidate weight from the quantile.  The scalar atom rule replaces input-specific atom ratios by marginal atom ratios.  The scalar score+ratio rule also uses the scalar mixed predictive score induced by this approximation.

The main finding is that the approximation gap is small for marginal coverage in this controlled setting.  Even at $20\%$ censoring per side and $\beta=1.2$, all four variants remain between $0.901$ and $0.905$ marginal coverage.  The difference appears more clearly in component behavior: as censoring becomes stronger, interior coverage drops and the high-shift atom becomes nearly always covered.  Thus scalar atom weights and calibration-only thresholds are empirically useful approximations, but they should not be interpreted as component-wise validity.

\begin{table}[t]
\centering
\caption{Experiment 2: approximation gap for positive shift ($\beta=1.2$).}
\label{tab:exp2-approx}
\small
\setlength{\tabcolsep}{4pt}
\begin{tabular}{llcccc}
\toprule
Censoring & Method & Coverage & Size & Interior cov. & High-atom cov. \\
\midrule
0.05 & WT exact & 0.900 & 2.132 & 0.815 & 0.979 \\
0.05 & calib-only exact ratio & 0.897 & 2.120 & 0.812 & 0.976 \\
0.05 & scalar atom ratio & 0.903 & 2.160 & 0.822 & 0.978 \\
0.05 & scalar score+ratio & 0.901 & 2.159 & 0.803 & 0.991 \\
\midrule
0.10 & WT exact & 0.902 & 1.728 & 0.731 & 0.990 \\
0.10 & calib-only exact ratio & 0.900 & 1.724 & 0.728 & 0.989 \\
0.10 & scalar atom ratio & 0.905 & 1.752 & 0.740 & 0.990 \\
0.10 & scalar score+ratio & 0.902 & 1.750 & 0.716 & 0.998 \\
\midrule
0.20 & WT exact & 0.902 & 1.283 & 0.530 & 0.995 \\
0.20 & calib-only exact ratio & 0.901 & 1.280 & 0.527 & 0.995 \\
0.20 & scalar atom ratio & 0.905 & 1.298 & 0.547 & 0.996 \\
0.20 & scalar score+ratio & 0.901 & 1.288 & 0.506 & 1.000 \\
\bottomrule
\end{tabular}
\end{table}

\subsection{Experiment 3: mixed-HDR geometry}\label{sec:exp3}

\Cref{tab:exp3-geometry} is the most censoring-specific diagnostic.  It records the structural form of the WT prediction set for positive shift.  With weak censoring, an ordinary interior interval $I$ occurs often.  As censoring becomes severe, atom-containing sets dominate.  At $20\%$ censoring per side, WT exact produces the atom-only upper set $\{U\}$ in $57.3\%$ of evaluated test inputs, and the form $I+U$ in $40.1\%$.  WT scalar exhibits the same qualitative behavior.  No disconnected set appears in any cell, consistent with its non-genericity for a unimodal tilted-Gaussian interior (\Cref{rmk:disconnected}); the $\{U\}$ column is the empty-interior atom-only regime of \Cref{prop:structure}, which the corrected structural characterization now names explicitly.

This experiment demonstrates why the censored problem should not be presented as a mechanical Gaussian interval extension.  The output is a mixed-HDR set on $\{L\}\cup(L,U)\cup\{U\}$, and its geometry changes with the censoring regime and target shift.

\begin{table}[t]
\centering
\caption{Experiment 3: mixed-HDR geometry for positive shift ($\beta=1.2$). Entries are frequencies of set forms.}
\label{tab:exp3-geometry}
\small
\setlength{\tabcolsep}{5pt}
\begin{tabular}{lccccc}
\toprule
Censoring & Method & $I$ & $L+I$ & $I+U$ & $U$ \\
\midrule
0.01 & WT exact & 0.358 & 0.000 & 0.634 & 0.009 \\
0.01 & WT scalar & 0.355 & 0.000 & 0.644 & 0.001 \\
\midrule
0.05 & WT exact & 0.134 & 0.000 & 0.764 & 0.102 \\
0.05 & WT scalar & 0.090 & 0.000 & 0.864 & 0.045 \\
\midrule
0.10 & WT exact & 0.064 & 0.000 & 0.686 & 0.249 \\
0.10 & WT scalar & 0.027 & 0.001 & 0.808 & 0.163 \\
\midrule
0.20 & WT exact & 0.024 & 0.002 & 0.401 & 0.573 \\
0.20 & WT scalar & 0.001 & 0.003 & 0.467 & 0.529 \\
\bottomrule
\end{tabular}
\end{table}

\subsection{Experiment 4: density-ratio estimation}\label{sec:exp4}

The preceding experiments used known ratios.  \Cref{tab:exp4-beta-est} tests practical estimation of the interior tilt parameter $\beta$.  Two estimators that ignore the residual shift $\beta b$---point pseudo-labels and source-predictive sampling---recover only $\beta\,a/(a+b)$ of the true slope, here exactly half since $a=\Var_s(\bhmu)=\|\theta\|^2=1$ and $b=\bar\sigma_s^2=\sigeps^2=1$; the attenuated slope ($\widehat\beta\approx\pm0.6$) produces conservative, oversized sets (coverage $\approx0.955$, size $\approx2.35$).  The moment-matching initializer and tilted predictive sampling both recover the full slope and nearly match the known-$\beta$ rule.

Two caveats temper this table.  First, it separates the attenuated estimators (point, source-predictive) from the consistent ones (moment matching, tilted predictive), but it does \emph{not} separate moment matching from tilted predictive sampling.  In this linear-Gaussian oracle, $\bar\mu_t-\bar\mu_s=\beta\|\theta\|^2$ and $\widehat v_\mu^2=\|\theta\|^2$, so moment matching returns $\beta$ exactly and the residual correction that tilted sampling supplies is not needed.  The theoretical advantage of tilted predictive sampling---restoring the residual shift $\beta b$ that mean-based matching omits (\Cref{prop:lrbias})---becomes visible only when the predictable component $\bhmu(X)$ underestimates the explained variance, i.e.\ when the backbone is imperfect or the conditional mean is nonlinear, a regime this oracle diagnostic does not exercise.  We therefore report moment matching and tilted predictive sampling as comparable here, retain tilted predictive sampling as the principled default on the theoretical grounds of \Cref{sec:dre}, and defer a discriminating comparison to a fitted-head study (\Cref{sec:conclusion}).  Second, the atom ratios are taken as known censoring fractions throughout; estimating them jointly with $\beta$ under input-only target data is left to future work.

\begin{table}[t]
\centering
\caption{Experiment 4: density-ratio parameter estimation.}
\label{tab:exp4-beta-est}
\small
\setlength{\tabcolsep}{4pt}
\begin{tabular}{llcccc}
\toprule
$\beta$ & Estimator & $\widehat\beta$ & Abs. error & Coverage & Size \\
\midrule
-1.2 & known & -1.200 & 0.000 & 0.901 & 1.756 \\
-1.2 & moment matching & -1.199 & 0.025 & 0.899 & 1.740 \\
-1.2 & point pseudo-label & -0.599 & 0.601 & 0.954 & 2.344 \\
-1.2 & source predictive & -0.600 & 0.600 & 0.954 & 2.345 \\
-1.2 & tilted predictive & -1.189 & 0.031 & 0.903 & 1.764 \\
\midrule
1.2 & known & 1.200 & 0.000 & 0.901 & 1.745 \\
1.2 & moment matching & 1.199 & 0.024 & 0.900 & 1.747 \\
1.2 & point pseudo-label & 0.599 & 0.601 & 0.955 & 2.345 \\
1.2 & source predictive & 0.600 & 0.600 & 0.955 & 2.356 \\
1.2 & tilted predictive & 1.190 & 0.028 & 0.901 & 1.752 \\
\bottomrule
\end{tabular}
\end{table}

\noindent\textbf{Summary of empirical findings.} The experiments support four claims.  First, the exact candidate-weighted mixed-space rule delivers the intended marginal coverage under latent label shift.  Second, target-score tilting is essential for efficiency.  WT is much shorter than WS when both have valid weighting.  Third, scalar atom weights and calibration-only thresholds can be accurate for marginal coverage, but component-wise coverage may still be uneven because a single pooled threshold is used.  Fourth, the censored setting produces genuinely mixed prediction sets, often including boundary atoms or even atom-only forms under severe censoring.  A real-data study with censored molecular endpoints remains an important next step, but the synthetic section now matches the revised theoretical claims.

\section{Conclusion}\label{sec:conclusion}

This paper developed conformal Bayes for two-sided censored Gaussian regression under label shift. 
The censoring-specific contributions are the mixed posterior predictive distribution with atoms at the detection limits, 
the three-term closed-form normalizer $Z(x)$, the mixed-space score that unifies the two atom components 
and the interior density, and the closed-form mixed-HDR characterization, which ranges from an interior interval 
to a boundary atom alone.  Validity and efficiency are cleanly separated.  
Exact oracle validity is model-free but requires a valid mixed-space importance identity 
and the candidate-weighted rank, whereas efficiency is a population mixed-HDR property of the target score.  
The orthogonality of measure weighting and score tilting carries over to the mixed space.  
The two genuinely censoring-induced subtleties are the latent-to-mixed-space shift \Cref{rmk:latent_observed_shift} 
and the marginal atom-weight approximation \Cref{rmk:marginalapprox}, both of which are stated explicitly 
rather than hidden in implementation details.

Several limitations remain.  The method relies on a Laplace-approximated censored posterior, 
whose approximation error may grow near the censoring limits.  The two-sided Tobit head assumes Gaussian latent errors, 
and richer likelihoods such as skew-$t$ or Weibull models would require numerical estimation of $Z(x)$.  
The experiments are controlled synthetic diagnostics using an oracle Gaussian predictive model, 
so a fitted-head study and real censored molecular data remain important next steps.  
If non-Gaussian or mixture predictive heads are used, disconnected mixed sets may arise 
and will require a clear communication protocol for domain scientists.

Future work includes full censored conformal Bayes with leave-one-out calibration, 
parameter-posterior correction through reweighted fine-tuning of the two-sided Tobit head, 
and non-Gaussian censored likelihoods, where the tilting identity and method hierarchy persist but $Z(x)$ is estimated numerically.

\bibliographystyle{abbrvnat}
\bibliography{sjc}

\clearpage
\appendix

\section{Gradient, Mills-ratio derivative, and Hessian}\label{app:hessian}

This appendix gives the calculus behind the Laplace approximation used for the two-sided Tobit head in \Cref{sec:model}.  We derive the score and negative Hessian of the observed-data log-likelihood, show how the Mills-ratio derivative enters the censored terms, and justify the concavity statement used to define the Gaussian approximation around the MAP.

With $\eta_i^L=(L-\theta^\top\backbone(x_i))/\sigeps$, $\eta_i^U=(U-\theta^\top\backbone(x_i))/\sigeps$ and $\lambda(t)=\phist(t)/\Phi(t)$, the chain rule gives $\nabla_\theta\eta_i^L=-\backbone(x_i)/\sigeps$ and $\nabla_\theta(-\eta_i^U)=\backbone(x_i)/\sigeps$, so
\[
\nabla_\theta\ell=-\frac1\sigeps\!\sum_{\wt{y}_i=L}\!\lambda(\eta_i^L)\backbone(x_i)+\frac1{\sigeps^2}\!\sum_{\wt{y}_i\in(L,U)}\!(\wt{y}_i-\theta^\top\backbone(x_i))\backbone(x_i)+\frac1\sigeps\!\sum_{\wt{y}_i=U}\!\lambda(-\eta_i^U)\backbone(x_i).
\]
Since $\phist'(t)=-t\phist(t)$ and $\Phi'(t)=\phist(t)$, the quotient rule gives $\lambda'(t)=-t\lambda(t)-\lambda(t)^2=-\lambda(t)(\lambda(t)+t)<0$ (because $\lambda$ is strictly decreasing). Differentiating each term yields the negative Hessian \eqref{eq:hessian}; each censored block contributes $-\lambda'(\eta)\backbone\backbone^\top/\sigeps^2\succeq0$ and each interior block $\backbone\backbone^\top/\sigeps^2\succeq0$. If the design $\{\backbone(x_i)\}$ has full column rank these sum to a positive-definite matrix, so $\ell$ is strictly concave; in any case the Gaussian prior adds $\tau^{-2}I_d\succ0$, making the penalized MAP objective strictly concave with a unique global maximizer.

\section{Latent-Pushforward Tilting Details}\label{app:latent-pushforward}

This appendix records the conditional latent-pushforward calculation that is only summarized in \Cref{sec:model}.  The main text uses the observed-space ratio $w(\wt y)=dP^t_{\wt Y}/dP^s_{\wt Y}(\wt y)$ to define the practical tilted predictive distribution.  Here we instead start from a latent ratio $w^*(y^*)$ and then censor the resulting tilted latent predictive distribution.  This derivation explains why exact atom weights are tail averages and why the scalar atom weights used in the main algorithm are an observed-space approximation.

The latent tilted predictive satisfies
\[
 p_t^{\rm lat}(y^*\mid x,\dtrain)=\frac{p_s(y^*\mid x,\dtrain)w^*(y^*)}{\int p_s(u\mid x,\dtrain)w^*(u)\,du}.
\]
Pushing this law forward through $T$ gives the left atom
\[
 p_t^{\rm lat}(\wt Y=L\mid x,\dtrain)
 =\frac{\int_{-\infty}^L p_s(y^*\mid x,\dtrain)w^*(y^*)\,dy^*}{\int p_s(u\mid x,\dtrain)w^*(u)\,du}
 =\frac{\pi_L(x)w_x(L)}{Z_x^{\rm lat}},
\]
where
\[
Z_x^{\rm lat}
=
\pi_L(x)w_x(L)
+
\int_L^U f_s(y\mid x)w^*(y)\,dy
+
\pi_U(x)w_x(U).
\]
The right atom is analogous.  On the interior $T$ is one-to-one, so the density is
\[
 \frac{f_s(\wt y\mid x)w^*(\wt y)}{Z_x^{\rm lat}},\qquad L<\wt y<U.
\]
This derivation gives the exact conditional latent pushforward construction.  The practical method in the main text replaces $w_x(L),w_x(U)$ by the scalar marginal ratios in \eqref{eq:scalaratoms} and uses the observed-space tilted predictive distribution \eqref{eq:predictive-tilt-general}.  This replacement is exact only under the observed-space identity \eqref{eq:obsimportance} or an equivalent condition making the tail-averaged ratio independent of $x$.

\section{Proof of \Cref{prop:structure} (Structural Characterization)}\label{app:proof-structure}

This appendix proves that the WT mixed-HDR set is the union of two independently thresholded atoms and one interior super-level interval.  This is the precise structural statement that replaces the earlier fixed enumeration of possible forms.

\begin{proof}
Atom membership is a direct threshold comparison on the $\mu$-density: $L\in\cC^{\mathrm{WT}}(x)\iff\piLt(x)\ge\lambda^*$ and $U\in\cC^{\mathrm{WT}}(x)\iff\piUt(x)\ge\lambda^*$, which define $A_L,A_U$. On $(L,U)$ the tilted density $\ft\propto\phist((\wt{y}-\mut)/\bhsigma)$ is a Gaussian restricted to the interval and is therefore unimodal, so its super-level set $I(x)=\{\wt{y}\in(L,U):\ft\ge\lambda^*\}$ is convex, i.e.\ empty or a single interval. When nonempty, the defining inequality is $(\wt{y}-\mut)^2/(2\bhsigma^2)\le\beta\bhmu+\tfrac12\beta^2\bhsigma^2-\log(Z(x)Z_w\bhsigma\sqrt{2\pi})-\log\lambda^*$, which yields $r(x)$ and $[y_\ell,y_u]$ as in \eqref{eq:interval}; it is empty exactly when the right-hand side is negative, i.e.\ when $\lambda^*$ exceeds the interior maximum. Since $A_L$, $A_U$, and $I$ are determined by three independent thresholdings, $\cC^{\mathrm{WT}}=A_L\cup I\cup A_U$ ranges over precisely the unions listed in (a)--(b). Within (a), the union is connected iff every included atom adjoins the interval (an included $L$ requires $y_\ell=L$, equivalently $\ft(L^+)\ge\lambda^*$; symmetrically at $U$), and disconnected otherwise. This accounts for every combination, including atom-only sets and the empty set when $I=\emptyset$. No claim of a fixed number of forms is made: the partition is by the three independent decisions, not by an enumerated list.
\end{proof}

\section{Proofs of \Cref{lem:varcomp} and \Cref{prop:lrbias}}\label{app:proof-lrbias}

This appendix collects the short calculations behind the pseudo-label bias discussion in \Cref{sec:dre}.  The first proof explains why point pseudo-labels remove within-input predictive variance.  The second proof derives the attenuation factor for source-predictive pseudo-label sampling.

\begin{proof}[\Cref{lem:varcomp}]
By the law of total variance, $\Var_s(Y^*)=\Var_{\ps(x)}\{\E[Y^*\mid X]\}+\E_{\ps(x)}\{\Var(Y^*\mid X)\}$; substituting $\E[Y^*\mid x]\approx\bhmu(x)$, $\Var(Y^*\mid x)\approx\bhsigma^2(x)$ gives $\Var_{\ps(x)}[\bhmu(X)]\approx v_s^2-\bar\sigma_s^2$.
\end{proof}
\begin{proof}[\Cref{prop:lrbias}]
Tilting a Gaussian component shifts its mean by $\beta\times$ its variance. With $M\perp\varepsilon$, tilting by $e^{\beta(M+\varepsilon)}$ shifts $M$ by $\beta a$ and $\varepsilon$ by $\beta b$. Point pseudo-labels keep $M_t\sim\cN(\beta a,a)$; source-PS adds a mean-zero variance-$b$ residual, giving $\cN(\beta a,a+b)$. Comparing with the source label marginal $\cN(0,a+b)$ yields slope $\beta a/(a+b)$.
\end{proof}

\section{Proof of \Cref{prop:hierarchy} (Correction Hierarchy)}\label{app:proof-hierarchy}

This appendix proves the validity hierarchy for UT, WS, and WT.  The proof separates the calibration effect of the importance weights from the geometric effect of the chosen score.

\begin{proof}
\emph{(i)} UT thresholds the source quantile $\bq_{\mathrm{UT}}$ of $\mathcal L_{\ps}(s^{\mathrm{WT}})$ but is evaluated under $\mathcal L_{\pt}(s^{\mathrm{WT}})$, so the coverage error is $|F_t(\bq_{\mathrm{UT}})-F_s(\bq_{\mathrm{UT}})|\le\dTV(\mathcal L_{\pt}(s^{\mathrm{WT}}),\mathcal L_{\ps}(s^{\mathrm{WT}}))+O(n^{-1/2})$; direction depends on how shift changes atom masses and interior density.
\emph{(ii)} By the observed-space identity \eqref{eq:obsimportance}, $\E_{\ps}\{w(\wt{Y})f(X,\wt{Y})\}=\E_{\pt}\{f(X,\wt{Y})\}$ for all measurable $f$---the only ingredient for weighted conformal validity. Applying the weighted rank argument to the $n{+}1$ augmented sequence (with the test weight $w(\wt{y}_{\mathrm{test}})$) gives $\Prob_{\pt}(\wt{Y}_{\mathrm{test}}\in\cC^r_{\mathrm{exact}})\ge1-\alpha$ for any score, including $s^{\mathrm{WS}},s^{\mathrm{WT}}$; no Tobit/Laplace property is used.
\emph{(iii)} Dropping the single candidate term from the weighted CDF changes it by at most $w(\wt{y})/(\sum_i w(\wt{y}_i)+w(\wt{y}))$; the supremum over candidates is $\epsilon_{\mathrm{cand}}=O_{\Prob}(1/n)$ under bounded, nondegenerate weights.
\end{proof}

\section{Proof of \Cref{prop:NP} (Mixed-HDR Optimality)}\label{app:proof-NP}

This appendix gives the population optimality argument for mixed-HDR sets under the adopted mixed measure.  The result is conditional on this size convention, which charges one unit for each atom and Lebesgue length for the interior.

\begin{proof}
Minimize $\E_X[\mu(C(X))]$ subject to $\E_X\!\int_{C(X)}p_t(\wt{y}\mid X)\,d\mu\ge1-\alpha$. With multiplier $\eta\ge0$, the per-point Lagrangian contribution is $1-\eta\,p_t(\wt{y}\mid x)$, so a point enters iff $p_t(\wt{y}\mid x)\ge1/\eta=:\lambda^*$. The argument is measure-theoretic and unchanged on a mixed space because atoms and interior share the dominating measure $\mu$. For WT, $\{s^{\mathrm{WT}}\le\bq_{\mathrm{WT}}\}$ is the level set with $\lambda^*=e^{-\bq_{\mathrm{WT}}}$; WS thresholds $p_s$ and targets a source HDR, generally suboptimal for $\E[\mu(C)]$ under $p_t$ when $p_s\neq p_t$.
\end{proof}

\section{SCB vs.\ SCB-C}\label{app:comparison}

This appendix summarizes how the censored method differs from the uncensored split conformal Bayes construction.  The main differences are the mixed observed outcome space, the three-component label ratio, and the atom-plus-interior geometry of the prediction set.

\begin{table}[h]
\centering
\caption{Uncensored SCB \citep{Choi2026eiml} vs.\ two-sided censored SCB-C.}\label{tab:comparison}
\smallskip
\begin{tabular}{@{}lll@{}}
\toprule
Dimension & SCB (uncensored) & SCB-C (censored)\\
\midrule
Outcome space & $\R$ & $\{L\}\cup(L,U)\cup\{U\}$ (mixed)\\
Predictive & $\cN(\bhmu,\bhsigma^2)$ & atoms $+$ Gaussian interior\\
Tilted predictive & $\cN(\bhmu+\beta\bhsigma^2,\bhsigma^2)$ & reweighted atoms $+$ shifted trunc.\ Gaussian\\
$Z(x)$ & $e^{\beta\bhmu+\frac12\beta^2\bhsigma^2}/Z_w$ (1 term) & 3 terms (left atom, partial MGF, right atom)\\
Prediction set & symmetric interval & mixed HDR: atoms plus interior interval\\
Coverage focus & marginal coverage & marginal coverage on mixed observed space\\
Density ratio & 1 function $w(y)$ & 3 objects $w(L),w(\wt{y}),w(U)$\\
\bottomrule
\end{tabular}
\end{table}

\section{Commutativity of censoring and posterior averaging}\label{app:commute}

The atom mass can be computed by averaging the latent likelihood over the posterior and then censoring (``averaging first,'' used in the main text), or by censoring the likelihood and then averaging (``censoring first''); the two commute. Under the Laplace approximation $\pi(\theta\mid\dtrain)\approx\cN(\bhtheta,\bhSigma_\theta)$, averaging first gives the Gaussian latent predictive $\cN(\bhmu(x),\bhsigma^2(x))$ and $\Prob(\wt{Y}=L\mid x)=\Phi((L-\bhmu(x))/\bhsigma(x))$; censoring first gives $\E_\theta[\Phi((L-\theta^\top\backbone(x))/\sigeps)]$, which by the Gaussian-CDF convolution identity equals the same $\Phi((L-\bhmu(x))/\bhsigma(x))$. Under $S$ MCMC draws the equivalence persists: both yield $\frac1S\sum_s\Phi((L-\theta^{(s)\top}\backbone(x))/\sigeps)$, because the CDF of a Gaussian mixture is the mixture of component CDFs. The practical pitfall is summarizing a non-Gaussian MCMC posterior by its mean $\bar\theta$ and plugging into a single $\Phi$: passing $\E[\theta]$ through the nonlinear $\Phi$ violates Jensen's inequality and discards posterior skew/heavy tails. With MCMC, the ``censoring first'' (per-sample) average is therefore preferred; under the Gaussian Laplace posterior used here, ``averaging first'' is exact.

\end{document}